\documentclass[a4paper,11pt]{article}
\usepackage[lined,boxed,commentsnumbered,ruled,vlined,noend,
linesnumbered,boxed]{algorithm2e}
\usepackage{graphicx,amssymb,amsmath,amsthm,textcomp}
\input{psfig.sty}
\newtheorem{theorem}{Theorem}
\newtheorem{lemma}{Lemma}

\newtheorem{defn}{Definition}
\usepackage{color}

\newcommand{\IR}{\mathbb{R}}

\newtheorem{observation}{Observation}

\usepackage{mdframed}
\usepackage{fancybox}
\usepackage[x11names,dvipsnames,table]{xcolor} 
\usepackage{colortbl}
\usepackage{undertilde}
\usepackage{IEEEtrantools}

\newcommand{\remove}[1]{}

\newcommand{\probname}[1]{{\color{blue!80!black}{\textbf{{#1}}}}}

\newcommand{\blue}[1]{{\textcolor{blue}{#1}}}

\usepackage{mdframed}
\usepackage{lipsum}

\newmdenv[
  backgroundcolor=gray!15,
  topline=false,
  bottomline=false,
  skipabove=\topsep,
  skipbelow=\topsep,
  leftmargin=-5pt,
  rightmargin=-5pt,
  innertopmargin=3pt,
  innerbottommargin=3pt
]{siderules}

\title{Range Assignment of Base-Stations Maximizing Coverage Area without Interference\footnote{A preliminary version of the paper is to appear  in the proceedings of  29th Canadian Conference on Computational Geometry, 2017.}}

\author{Ankush Acharyya\thanks{ACM Unit, Indian Statistical Institute, Kolkata, India.
	Email: {\tt \{ankush\_r, nandysc\}@isical.ac.in}.} \and
		Minati De\thanks{Department of CSA, Indian Institute of Science, Bangalore, India.
		Email: {\tt minati@iisc.ac.in}.}  \thanks{Supported by DST-INSPIRE Faculty
		Grant (IFA-14-ENG-75).}\and Subhas C. Nandy\footnotemark[2]\and Bodhayan Roy
		\thanks{Department of CS, Masaryk University, Brno, Czech Republic. Email: {\tt  b.roy@fi.muni.cz}}}

\begin{document}
\thispagestyle{empty}
\maketitle

\begin{abstract}
We study  the problem of assigning non-overlapping geometric objects 
centered at a given set of points  such that the sum of area 
covered by them is maximized. If the points  are placed on a straight-line  and the objects are disks, 
then the problem is solvable in polynomial time. However, we show that the problem is 
 NP-hard  even for simplest  objects like disks or squares in ${\IR}^2$. Eppstein [CCCG, pages 260--265, 2016] 
proposed a polynomial time algorithm for maximizing the sum of radii (or 
perimeter) of non-overlapping balls or disks when the points are 
arbitrarily placed on a plane. We show that Eppstein's algorithm for 
maximizing sum of perimeter of the disks in ${\IR}^2$ gives a 
$2$-approximation solution for the sum of area maximization problem.
We propose a PTAS for our problem. These approximation results are 
extendible to higher dimensions. All these approximation results hold for 
the area maximization problem by regular convex  polygons with even number 
of edges centered at the given points.  
\end{abstract}

{\bf Keywords: } Quadratic programming, discrete packing, range assignment 
in wireless communication, NP-hardness, approximation algorithm, PTAS.

\vspace{-0.1in}
\section{Introduction}\label{intro}
Geometric packing problem is an important area of research in computational 
geometry, and it has wide applications in cartography, sensor network, wireless 
communication, to name a few. In the disk packing problem, the objective is 
to place maximum number of congruent disks (of a given radius) in a given 
region. Toth 1940 \cite{chang,Toth04} first gave a complete proof that hexagonal 
lattice packing produces the densest of all possible disk packings of both 
regular and irregular regions. Several variations of this problem are possible 
depending on various applications \cite{BentzCR13,Toth04}. 
\remove{
Another important 
variation in this research is as follows  \cite{ChanG14}: given a geometric 
range space $({\cal X}, {\cal S})$, 
\begin{description}
\item[pack-points problem:] find a maximum cardinality subset $Y \subseteq 
{\cal X}$ of points, no two of which are contained in a single region of 
${\cal S}$, and 
\item[pack-regions problem:] find a maximum-cardinality subfamily $C 
\subseteq {\cal S}$ of regions, no two of which intersect at a point in 
${\cal X}$.
\end{description}
Here the members of $\cal S$ may be geometric objects of some specific 
type (disk or rectangle or square of same or different sizes), and 
${\cal X}$ may be a continuous region or a set of points in a given 
region. The {\it pack-points problem} and {\it pack-regions problem} 
can be thought of as the “dual” problems associated to the set cover 
and hitting set problems for geometric range spaces. Here by dual, we 
mean the packing–covering duality or linear programming duality. The 
pack-regions problem sometimes referred to as the discrete 
independent set problem. Chan \cite{ChanG14} proved APX-hardness results 
for several instances of pack-points problem and pack-regions problem
when the members in $\cal S$ are axis-aligned rectangles, slabs in the 
plane and unit balls in ${\IR}^3$, and $\cal X$ consists of all points 
in the given region.}
In this paper, we will consider the following variation of the packing 
problem:


\begin{siderules}
\noindent\probname{Maximum area discrete packing (MADP):} Given a set of points $P$ = $\{p_1$, $p_2,
\ldots, p_n\}$ in ${\IR}^2$, compute the radii of a set of non-overlapping 
disks ${\cal C}=\{C_1, C_2,\ldots, C_n\}$, where $C_i$ is centered at 
$p_i \in P$, such that $\sum_{i=1}^n \text{area}(C_i)$ is maximum. 
\end{siderules}

The problem can be formulated as a quadratic programming problem as follows. Let 
$r_i$ be the radius of the disk $C_i$. Our objective is: 
\begin{siderules}
\begin{tabbing}
\= Subject to \= \kill
\> Maximize \> $\sum_{i=1}^n r_i^2$ \\
\> Subject to\> $r_i+r_j \leq \text{dist}(p_i,p_j)$, $\forall$ $p_i,p_j \in P$, $i \neq j$. 
\end{tabbing}
\end{siderules}
Here, $\text{dist}(p_i,p_j)$ denotes the Euclidean distance of $p_i$ and $p_j$.
The motivation of the problem stems from the range assignment problem in wireless networks.
Here the inputs are the base-stations. Each base-station is assigned with a range, and it covers
a circular area centered at that base-station with radius equal to its assigned range. The
objective is to maximize the area coverage by these base-stations without any interference.
In other words, the area covered by two different base-stations should not overlap. 
Surprisingly, to the best of our knowledge, there is no literature for the MADP problem.
A related problem, namely \blue{{\it maximum perimeter discrete packing} ({\bf MPDP})} 
problem, is studied recently by Eppstein \cite{eppstein}, where the objective is to 
compute the radii of the disks in $\cal C$ maximizing $\sum_{i=1}^n r_i$ 
subject to the same set of linear constraints. This is a linear programming 
problem for which polynomial time algorithm exists \cite{Papadimitriou}. In 
particular, here each constraint consists of only two variables, and such 
a linear programming problem can be solved in $O(mn^3\log m)$ time 
\cite{Megiddo83}, where $n$ and $m$ are number of variables and number of 
constraints respectively. In \cite{eppstein}, a graph-theoretic formulation 
of the MPDP problem is suggested. Let $G=(V,E)$ be a complete graph whose vertices $V$ correspond 
to the points in $P$; the weight of edge $(i,j) \in E$ ($i \neq j$) is
$\text{dist}(p_i,p_j)$, which corresponds to the constraint $r_i + r_j \leq 
\text{dist}(p_i,p_j)$. They computed the minimum weight cycle cover of $G$ in 
time $O(mn+n^2\log n)$ time. Since $m=O(n^2)$ in our case, the time complexity 
of this algorithm is $O(n^3)$. They further considered the fact that a constraint 
$r_i + r_j \leq \text{dist}(p_i,p_j)$ is useful if $\delta(p_i)+\delta(p_j) \geq 
\text{dist}(p_i,p_j)$, where $\delta(p)$ is the distance of the point $p$ and its
nearest neighbor in $P$; otherwise that constraint is redundant. They also  
showed that the number of useful constraints is $O(n)$, and thus the overall time 
complexity becomes $O(n^2\log n)$. They used further graph structure to reduce 
the time complexity. In ${\IR}^d$, the time complexity of this problem is 
shown to be $O(n^{2-\frac{1}{d}})$.

It is well-known that if $Q$ is a positive definite matrix, then the quadratic programming
problem which minimizes $\tilde{X}'Q\tilde{X}$ subject to a set of linear constraints 
$A\tilde{X} \leq \tilde{b}$, $\tilde{X} \geq 0$  is solvable in polynomial time \cite{Khachiyan}. However, if we present our
maximization problem as a minimization problem, the diagonal entries of the 
matrix $Q$ are all $-1$ and the off-diagonal entries are all zero. Thus, all 
the eigen values of the matrix $Q$ are $-1$. It is already proved that the 
quadratic programming problem is NP-hard when at least one of the eigen 
values of the matrix $Q$ is negative \cite{PardalosV91}. This indicates 
that the 
MADP problem also seems to be computationally hard. For the minimization version of an NP-hard quadratic 
programming with $n$ variables and $m$ constraints, an 
$(1-\frac{1-\epsilon}{(m(1+\epsilon))^2})$ factor approximation algorithm is proposed 
in \cite{FuLY98}, which works for all $\epsilon \in (0,1-\frac{1}{\sqrt{2}})$.
The time complexity of this algorithm  is $O(n^3(m\log \frac{1}{\delta}+\log\log\frac{1}{\epsilon}))$,
where $\delta$ is the radius of the largest ball inside 
the feasible region defined by the given  set of constraints.
For our MADP problem in $\IR^2$, a {\it 4-factor approximation algorithm} is easy to obtain. 
\begin{description}
\item[] For each 
point $p_i\in P$, let ${\cal N}(p_i) \in P$ be its nearest neighbor, and $\ell_i=\text{dist}(p_i,{\cal N}(p_i))$. We assign 
$r_i = \frac{1}{2}\ell_i$ for each $i \in\{1,2,\ldots, n\}$.
Thus, all the constraints are satisfied. The approximation factor follows from the
fact that in the optimum solution the radius $\rho_i$ of a disk centered at $p_i$ can take  value at most  $\ell_i$.
\end{description}

\subsection*{Our contribution}
In Section \ref{section2}, we first show that if the points in $P$ are placed on a straight line, then the MADP problem 
can be optimally solved in $O(n^2)$ time. In Section \ref{section3}, 
we show that MADP problem in ${\IR}^2$ is NP-hard. As a feasible solution of the MPDP problem is also a feasible solution
of the MADP problem, it is very natural to   ask whether an optimal solution of the MPDP problem is a good solution
for the MADP problem, or not. In Section \ref{section4}, we  answer this question in the affirmative. We  show that the
optimum solution for the MPDP problem proposed in \cite{eppstein} is a 2-approximation result for the MADP problem.
We also propose a PTAS for the MADP problem. In Section \ref{section6}, we show that the approximation results in
Sections \ref{section4} are extendible to higher dimensions. Finally, in Section \ref{section7} we show that all
these approximation results for the MADP problem in ${\IR}^2$ hold for any regular convex polygon with even number of edges.

\section{Preliminaries}
A solution for the MADP problem consists of disks with center at each point in $P$.
Their radii are all greater than or equal to zero\footnote{A disk with radius 0 implies that no disk is placed at that point.}.
A solution of the MADP problem is said to be \blue{{\it maximal}} if each disk touches 
some other disk (may be of radius 0) in the solution. From now onwards, by a \blue{{\it solution}} of a MADP problem, we will 
mean it to be a {\it maximal solution}.

\begin{figure}[h]
\centering\includegraphics[scale=0.45]{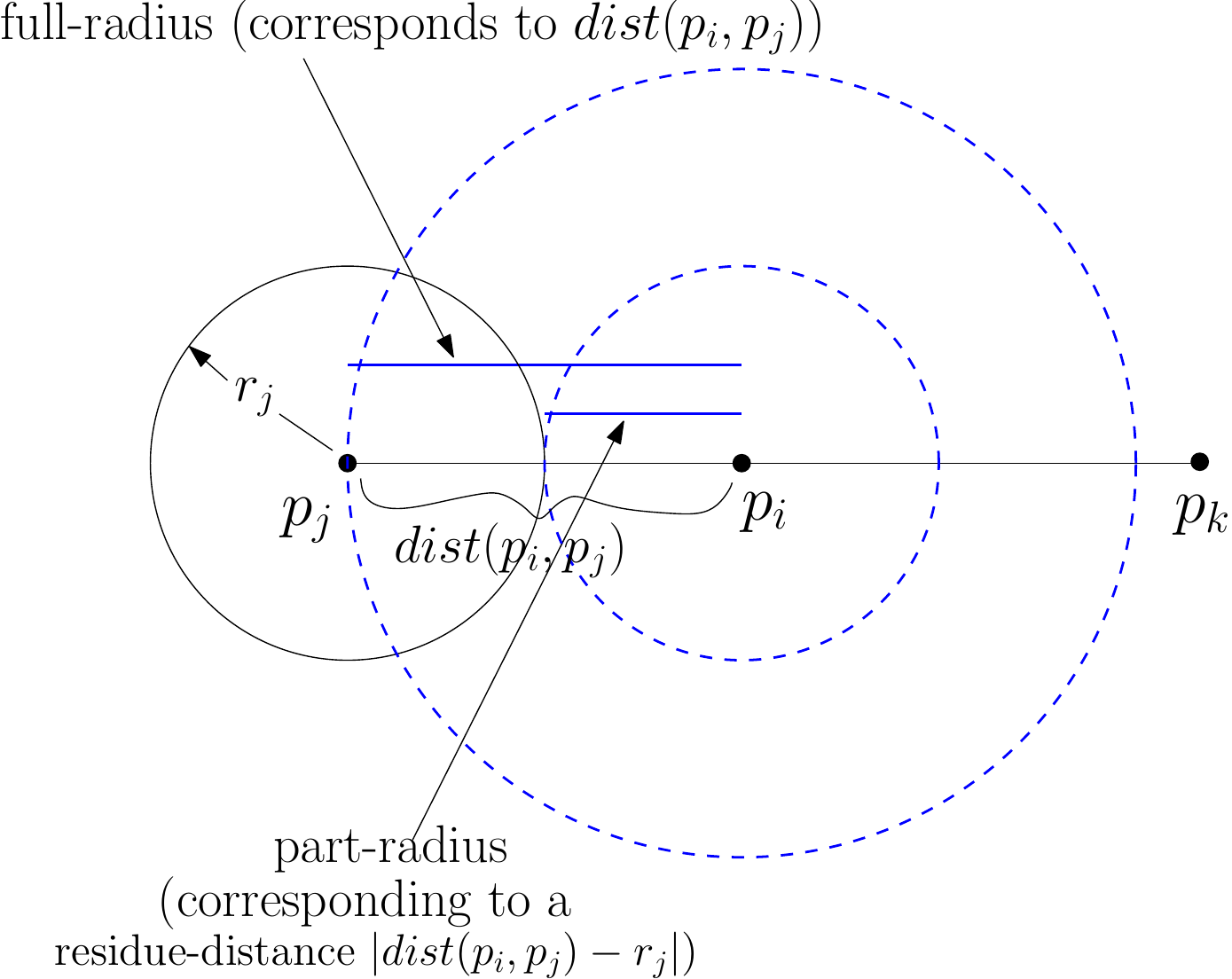}
\caption{full-radius, part-radius and residue-distance of $C_i$ with respect to $p_j$}
\label{twointervals}
\end{figure}

The nearest neighbor of a point $p_i \in P$ is denoted by ${\cal N}(p_i) \in P$.
Here, a point $p_i \in P$ is said to be a \blue{{\it defining point}} of the said
solution if it appears on the boundary of some disk in the solution; otherwise
it is said to be a \blue{{\it non-defining point}}. A {\it non-defining point}
$p_i\in P$ will be covered with a disk $C_i$ centered at point $p_i$, and
its radius $r_i$ is either {\em equal to} or {\em less 
than} $\text{dist}(p_i,q_i)$, where $q_i={\cal N}(p_i)$.
In the former case, $C_i$ is said to have \blue{{\it full-radius}}, and in the later case, 
$C_i$ is said to have \blue{{\it part-radius}} since the boundary of $C_i$ does not have 
any point in $P$. Let us consider a neighbor $p_j$ of the point $p_i$ which has a disk ${\cal C}_j$ of radius $r_j$.
We will use the term \blue{{\em residue-distance}} to indicate a feasible radius for the disk $C_i$ of length
$|\text{dist}(p_i,p_j) - r_j|$ for $i \neq j$, if $|\text{dist}(p_i,p_j) - r_j|\leq |\text{dist}(p_i,{\cal N}(p_i))|$
(see Figure \ref{twointervals}). Thus, the residue-distance of a disk $C_i$ (centered at 
$p_i$) is zero if ${\cal N}(p_i)$ is a {\em defining point}. For each full-radius (resp. part-radius)
of a disk $C_i$ corresponding to $p_i$, we define a \blue{{\it full-radius interval (resp. part-radius interval)}}
of length $2r_i$, where $r_i$ is the radius of $C_i$.

\section{MADP problem on a line} \label{section2} 
In this section, we will consider a constrained version of the MADP problem,
where the point set $P=\{p_1,p_2,\ldots,p_n\}$ lie on a given line $L$, 
which is assumed to be the $x \text{-} axis$. We also assume $\{p_1,p_2,\ldots,p_n\}$ is sorted
in left to right order. We use $d_i$ to denote the distance of the pair of points 
$(p_i,p_{i+1})$, $i=1,2, \ldots, n-1$. Our objective 
is to place non-overlapping disks centered at each point $p_i\in P$ such that the sum of the 
area formed by those disks is maximized. We will use $r_i$ to denote the radius of the disk centered at 
the point $p_i$, where $r_i \geq 0$ for $i=1,2,\ldots,n$.

\begin{lemma}\label{end_points}
In the optimum solution of the MADP problem on a line, at least one of the leftmost or rightmost 
point in $P$ must be either a {\em defining point} or its corresponding disk has {\em full radius}. 
\end{lemma}
\begin{proof}
For the contradiction, let the leftmost 
point $p_1$ in $P$ has radius $r_1$ satisfying $0 < r_1 < \text{dist}(p_1,{\cal N}(p_1))$ 
(see Figure \ref{zero_full}). If $r_2 = d_2<d_1-r_1$, then we can increase $r_1$, indicating 
the non-optimality of the solution. If $r_2=d_1-r_1$, then $r_3=\min(d_3, (d_2-(d_1-r_1)))$.
Assuming $r_3=d_2-(d_1-r_1)$ and proceeding similarly, we may reach one of the following two 
situations:
\begin{itemize}
 \item[1.]  $r_k=d_{k-1}-(d_{k-2}-( \ldots (d_1-r_1)))\ldots )$, and the values of
$r_{k+1},\ldots, r_n$ are independent of $r_1$.
\item[2.] $r_{n-1}=d_{n-2}-(d_{n-3}-( \ldots (d_1-r_1)))\ldots )$ and $r_n=d_{n-1}-r_{n-1}$.
\end{itemize}
Below, we show that in Case 1,  $S_k = r_1^2+r_2^2+\ldots+r_k^2$ can be increased while keeping the values of $r_{k+1}, \ldots, r_n$ unchanged.

\begin{tabbing}
xxx \= \kill
\>$S_k$ \= = \= $\pi\cdot(r_1^2+(d_1-r_1)^2+(d_2-(d_1-r_1))^2+\ldots + (d_k-(d_{k-1}-(\ldots(d_1-r_1))))^2)$\\
\>\> = \> $\pi\cdot (k\cdot r_1^2  -2r_1\cdot c_2 +  c_1 )$, \\
\>where \= $c_1$ = \= $d_1^2+ (d_2 -d_1)^2+\ldots +(d_k-(d_{k-1}-(\ldots + (-1)^k\cdot d_1)))^2$,\\
\> and \= $c_2=(d_1-(d_2 -d_1)+\ldots+ (-1)^{k-1}(d_k-(d_{k-1}-(\ldots + (-1)^k\cdot d_1))))).$
\end{tabbing}

Thus, $S_k$ is a parabolic function whose minimum is attained at $r_1=\frac{c_2}{k}$, and it attains maximum 
at the boundary values of the feasible region of $r_1$, i.e either at $r_1=0~\text{or}~d_1$.

In Case 2, if $r_n > r_{n-1}$, we can increase the sum $S_n$ by setting $r_n=d_{n-1}$, $r_{n-1}=0$ and 
keeping $r_1, r_2, \ldots, r_{n-2}$ unchanged. Now, $r_1^2+ r_2^2+ \ldots + r_{n-2}^2$ can further be increased
as in Case 1. Similarly, if $r_1 > r_2$ then also $S_n$ can be increased by setting $r_1=d_1$ and $r_2=0$,
and then maximizing $r_3^2+r_4^2+\ldots+r_n^2$ as in Case 1. If $r_n \leq r_{n-1}$ and $r_1 \leq r_2$,
then also $S_n$ is a parabolic function of $r_1$, and it is maximized at either $r_1=0$ or $r_1=
\min(d_1,\alpha)$ where $\alpha$ = value of $r_1$ for which $r_{n-1}=d_{n-1}$ \footnote{Here right-end 
of the feasible region of $r_1$ is obtained by placing a disk of radius $d_{n-1}$ at $p_n$,
and placing disks at points $p_{n-1}, \ldots, p_2$ touching those of $p_n, \ldots, p_3$, and
then placing the disk of radius $\alpha$ at $p_1$ that touches the disk at $p_2$. Here surely
$\alpha \leq d_1$.}. 
\end{proof}

\begin{figure}[h]
\centerline{\includegraphics[scale=0.76]{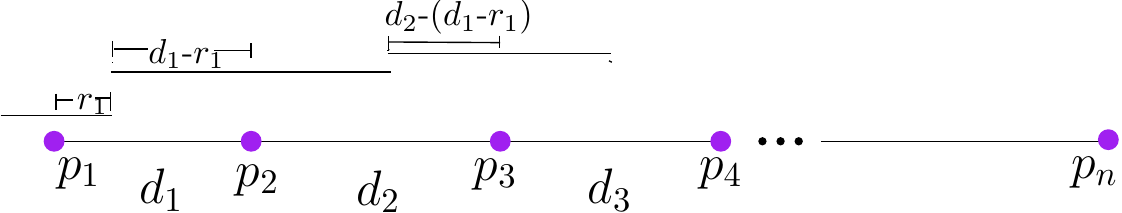}}
\caption{An instance, considering $k=3$}
\label{zero_full}
\end{figure}

Lemma \ref{end_points} says that in an optimum solution all the disks have either full-radius or zero radius or 
has radius equal to the residue distance with respect to the radius of its neighboring points.

Full-radius disks (intervals) are easy to get. For each point $p_i$, find its nearest
neighbor ${\cal N}(p_i) = p_{i-1} \text{ or}~p_{i+1}$, and define an interval
of length equal to  $2\cdot\text{dist}(p_i,{\cal N}(p_i))$, centered at $p_i$.
We now describe the generation of all possible part-radius intervals for each 
point $p_i \in P$ considering them in left to right order.

\begin{itemize}
\item For both the points $p_1$ and $p_2$, there is no part-radius interval.
\item If ${\cal N}(p_2) = p_1$, then for point $p_3$, there is a  part-radius interval of length $2(d_2-d_1)$,
centered at $p_3$; otherwise there is no 
part-radius interval for the point $p_3$.
\item In general, for an arbitrary point $p_k$ if there are $m$ number of part-radius intervals $I_1, I_2,
\ldots, I_m$ of lengths $2\delta_1, 2\delta_2, \ldots, 2\delta_m$ respectively, then 
each of these intervals $I_j$ gives birth to a part-radius interval for the point $p_{k+1}$ with center at
$p_{k+1}$ and of length $2(d_k-\delta_j)$. \\
In addition, if ${\cal N}(p_{k}) = p_{k-1}$, then for point $p_{k+1}$, there is another part-radius interval
centered at $p_{k+1}$ and of length $2(d_k-d_{k-1})$.
\end{itemize} 

Finally, we have ${\cal I}= \cup_{i=1}^n {\cal I}_i$. A similar process is performed to generate part-radius intervals ${\cal J}$ 
by considering the points in $P$ in
right to left order.

\begin{lemma}\label{no_is_i}
For a set $P$ of $n$ points lying on a line $L$, the maximum number of intervals generated 
 by the above procedure is $\Theta(n^2)$.
\end{lemma}
\begin{proof}
Let us first consider the forward pass as explained above. Here, for each point $p_i$ (in order) a 
full-radius interval is generated, and the full-radius interval for point $p_i$ may generate a 
part-radius interval for each point $p_j, j=i+1, \ldots, n$. Thus, for all the points in $P$, we 
may get $O(n^2)$ intervals. To justify the number of intervals is $\Omega(n^2)$, see the  
demonstration in Figure \ref{intervals}. Here the points $p_i=(x_i,0)$, $i=1, 2, \ldots, n$ are 
placed on the $x$-axis, where $x_1=0, x_2=1$ and $x_{i}=x_{i-1}+(x_{i-1}-x_{i-2})+0.5$, $i=3, 4, \ldots, 
n$. Here for each generated interval at $p_i$, a part-radius interval for the points $p_j, j=i+1, 
\ldots, n$ will be generated. The same argument follows for the reverse pass also.
\end{proof}

 \begin{figure}[h]
\centering\includegraphics[scale=0.5]{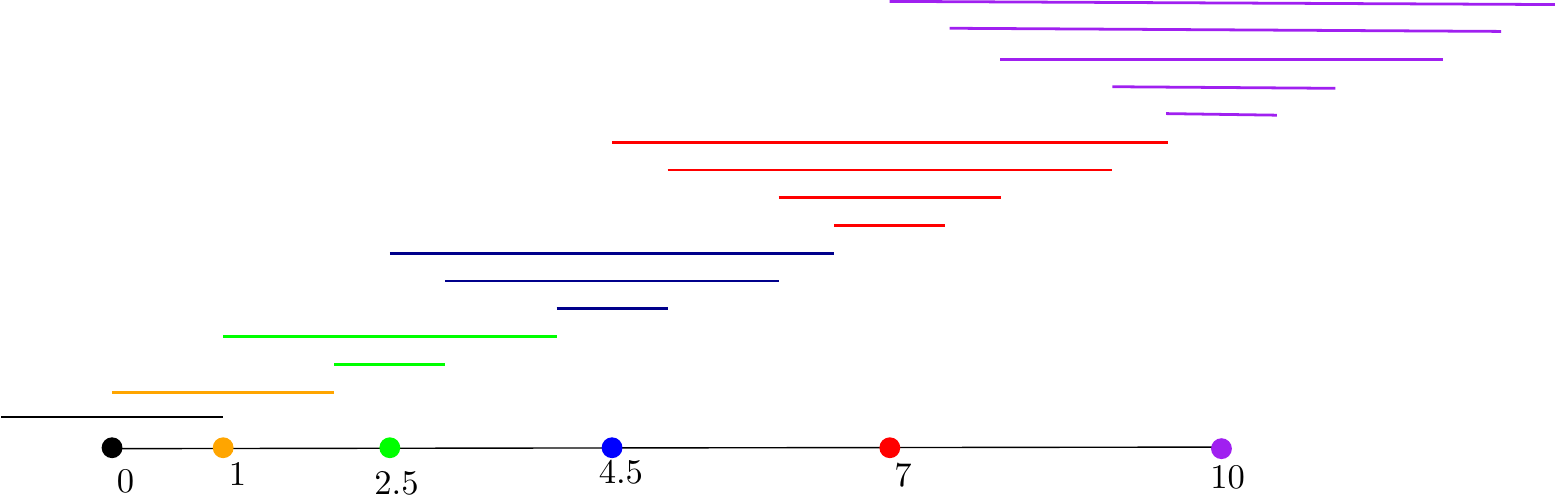}
\caption{An $\Omega(n^2)$ instance of full and part radius intervals }
\label{intervals}
\end{figure}

For each of these intervals we assign weight equal to the square of 
their half-length. We sort the right end points of these intervals. For this sorted set of weighted
intervals, we find the maximum weight independent set. This leads us to the following theorem.
\begin{theorem}
Given a set $P$ of $n$ points on a line $L$, one can place non-overlapping disks maximizing
  sum of their area in $O(n^2)$ time.
\end{theorem}
\begin{proof}
 We can generate the intervals in $O(n^2)$ time as follows. Given a set of intervals ${\cal I}_i$ 
 (of full- and part-radius) generated for a point $p_i$ which are sorted by their right end-points, 
 we can generate the set of part-radius intervals ${\cal I}_{i+1}$ for the point $p_{i+1}$ in $O(i)$ 
 time. Thus, total time for interval generation is $O(n^2)$ in the worst case. Since intervals for 
 each point $p_i$ are generated in sorted manner, ordering them with respect to  their end-points 
 also takes $O(n^2)$ time.  Finally, computing the maximum weight independent set of the sorted set 
 of intervals $\cup_{i+1}^n {\cal I}_i$ using dynamic programming  needs $O(n^2)$ time
 \cite{kleinberg2006algorithm}.
 
 The correctness of the algorithm follows from the fact that, if  there is an
 interval $\theta$ corresponding to point $p_i$ in the optimum solution that does not belong to ${\cal I} \cup {\cal J}$, then
 it is not generated by any interval in ${\cal I}_{i-1}$ and ${\cal J}_{i+1}$.  As a result it does not touch any interval of 
 ${\cal I}_{i+1}$ and also ${\cal J}_{i-1}$. Thus, interval $\theta$ can be elongated to increase the total covering area. 
 \end{proof}

\section{MADP problem in ${\IR}^2$ is NP-hard} \label{section3}
Here, we show  that the MADP problem in $\IR^2$ is NP-hard by a polynomial time reduction of planar rectilinear 
monotone 3-SAT (PRM-3SAT) problem to this problem.

\begin{defn} \label{def1}
A \blue{\it planar rectilinear monotone 3-SAT} (PRM-3SAT) is a 3-SAT formula $\theta$
 such that in every clause of $\theta$, either all the literals are positive, or all the literals are negative.
 Furthermore, $\theta$ has an embedding $\xi$ in $\IR^2$ with the following properties:
 
 \begin{itemize}
  \item[(i)] The variables and clauses of $\theta$ are represented in $\xi$ by axis parallel squares and rectangles respectively.
  \item[(ii)] All the squares representing the variables have the same size and they lie on the x-axis.
  \item[(iii)] All the rectangles representing the clauses have the same height. But their lengths may vary.
  \item[(iv)] The rectangles for positive clauses are above the x-axis while the rectangles for the negative clauses are below the x-axis.
  \item[(v)] The paths joining variables to their respective clauses are just vertical lines, called {\em clause-variable connecting path} ({\em CVC-path}, in short).
  \item[(vi)] The corners of the squares and rectangles representing the variables and clauses respectively, and 
  the end-points of all the {\em CVC-paths} in the embedding $\xi$ are latice points. 
 \end{itemize}
 \end{defn}
Given a PRM-3SAT formula $\theta$, its embedding $\xi$, as stated above, can be obtained in polynomial time.
In \cite{plane-rec}, it is shown that PRM-3SAT problem is NP-complete. 

 We call a set of non-intersecting disks centered at the given points, a \emph{disk configuration}. 
 A disk configuration that gives the maximum area is called a \emph{maximum disk configuration}.
 A disk is said to be \emph{on} a point if it is centered at that point.
  
 \remove{
\begin{figure}[h]
\centering\includegraphics[scale=0.45]{22.pdf}
\caption{Maximum disk configurations of (a) points at the vertices of a right-angled isosceles triangle, and 
 (b) points at the vertices of an equilateral triangle.}
\label{figexample}
\end{figure}
  }

   \begin{figure} 
\centering\includegraphics[scale=0.25]{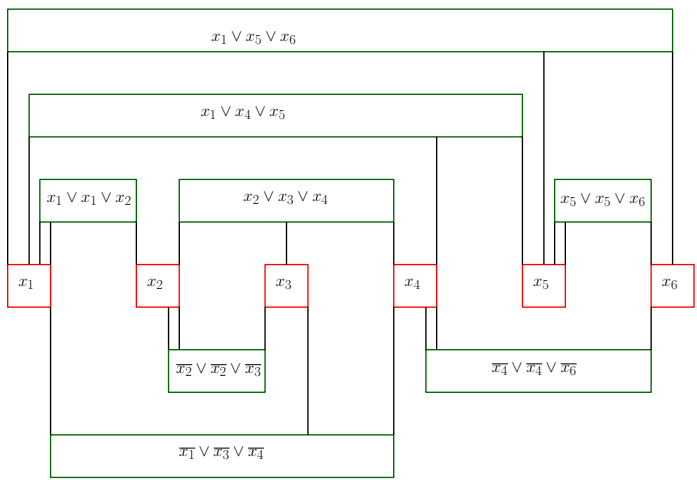}
\caption{A planar monotone rectilinear 3-SAT formula.}
\label{monsat}
\end{figure}

 \subsection{The reduction}

We start with an embedding of a planar monotone 
 rectilinear 3-SAT formula $\theta$, as in Definition \ref{def1}, with variables $\{x_1, x_2, \ldots, x_n \}$ and 
 clauses $\{C_1, C_2, \ldots, C_m \}$ 
 (see Figure \ref{monsat} for an example). Observe that, a clause with two literals can be made a three literal 
 clause by duplicating its any one of the literals. Thus, we can assume that all clauses in $\theta$ have three literals.
  We replace each clause with a
clause-gadget and each variable with a variable-gadget using point sets. Also we put points along the CVC-paths
connecting each clause with the variables in it.  For convenience we use points of three colors, namely red, green and blue, 
 in our reduction. Our configuration of points will 
 contain the following sub-configurations.

\remove{
The clause-gadgets particularly are quite involved, because we 
cannot use configurations such as in Figure \ref{figexample} either because they are too simple to serve our purpose or
have irrational coordinates.
We replace each {\em CVC-path} (joining a clause and a variables) with a sequence of points on which small disks can be drawn.
The positioning of these points ensures that a satisfying assignment of $\theta$ always leads to a disk configuration
that has more area than a disk configuration corresponding to a non-satisfying assignment of $\theta$.  
Hence observing the area of the maximum disk configuration we can determine whether $\theta$ is satisfiable or not.
}
 \begin{figure}[t] 
\centering\includegraphics[scale=0.35]{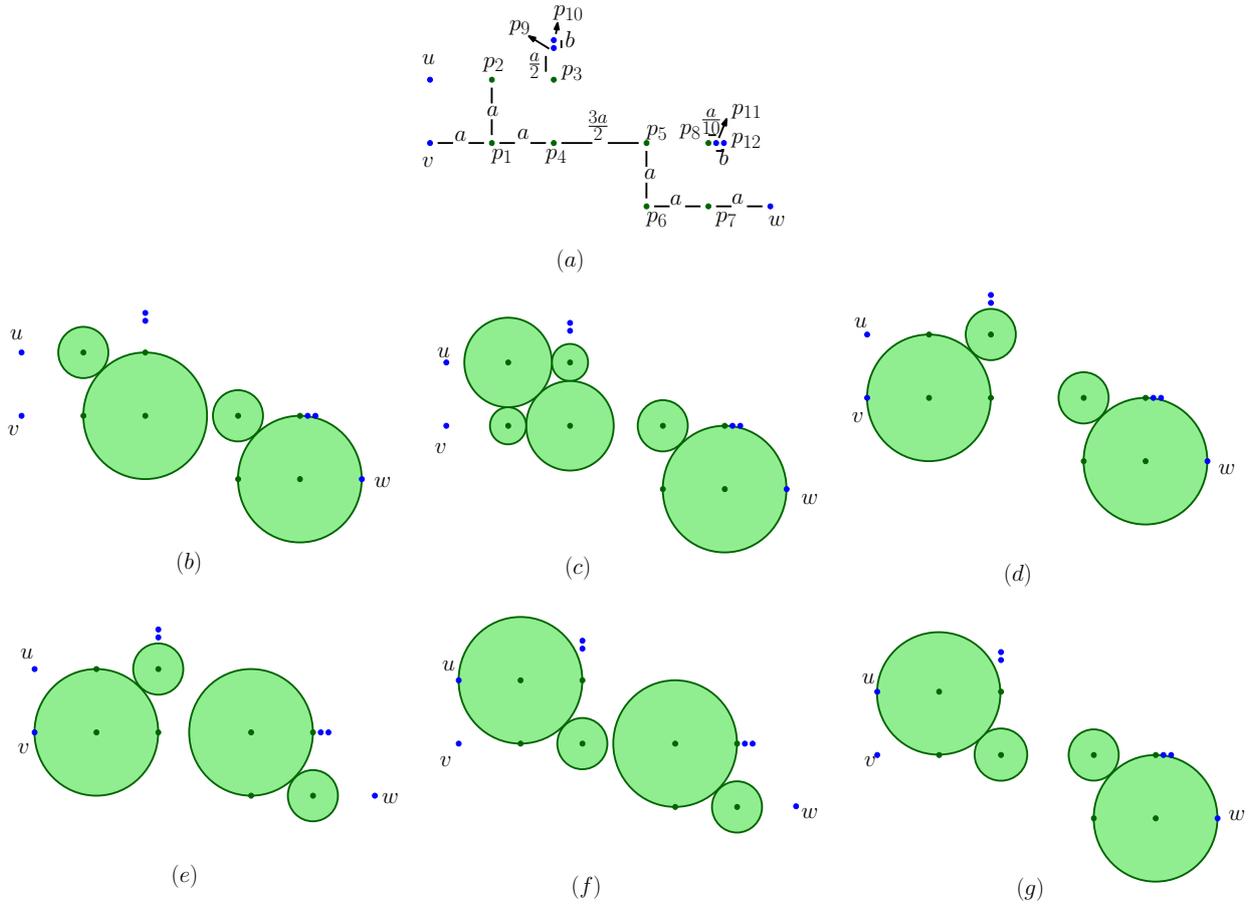}\caption{
(a) A clause-gadget.
(b,c) Two maximum disk configurations touching only $w$.
(d) A maximum disk configuration touching both $v$ and $w$.
(e) A maximum disk configuration touching only $v$.
(f) A maximum disk configuration touching only $u$.
(g) A maximum disk configuration touching only $u$ and $w$.
}
\label{figclpt}
 \end{figure}

\subsubsection*{Clause gadget} 
 A clause gadget corresponding to any clause $C_\alpha$ of $\theta$ has eight green and four blue points.  Let 
the coordinate of one green point is $p_1=(\mu,\nu)$.
The coordinates of the 
other seven green points are
$p_2=(\mu,\nu + a)$, $p_3=(\mu + a,\nu + a)$, $p_4=(\mu + a,\nu)$, 
$p_5=(\mu + \frac{5a}{2},\nu)$, $p_6=(\mu + \frac{5a}{2},\nu - a)$, $p_7=(\mu + \frac{7a}{2},\nu)$
and $p_8=(\mu + \frac{7a}{2},\nu - a)$. 
The coordinates of the blue points are $p_9=(\mu + a,\nu + \frac{3a}{2})$, $p_{10}=(\mu + a,\nu + \frac{3a}{2} + b)$,
$p_{11}=(\mu + \frac{7a}{2} + \frac{a}{10},\nu)$ and $p_{12}=(\mu + \frac{7a}{2} + \frac{a}{10} + b,\nu)$
(see Figure \ref{figclpt}(a)). Other than these points, there are three blue points $v_i$, $v_j$, and $v_k$, which are at a distance of $a$ 
units to the left of $p_1$, to the left of $p_2$, and to the right of $p_7$ respectively. 
These are the points on the {\em CVC-path} from the variable-gadgets $x_i$, $x_j$ and $x_k$ appearing in this clause. We choose $b$ and 
$a$ later depending on the number of variables $n$ and number of clauses $m$ of $\theta$. We have the following observation on clause-gadgets, 
which we will prove in Lemma \ref{lemclpt}.
 
\begin{observation}
  The total area of the disks centered at the points of a clause-gadget for a clause with three literals is maximized 
  only if there is a disk of radius $a$ at some green point of 
  that clause-gadget touching the {\em last}\footnote{By the {\it last blue point} of a variable 
 $x_i$ and a clause $C_\alpha$, we mean the blue point $v_i$ of the CVC-path
 closest to the clause-gadget of $C_\alpha$.} blue point of at least one {\em CVC-path} 
  reaching to that clause-gadget (see Figures \ref{figclpt}(b--f)).
\end{observation} 

\subsubsection*{Variable-gadget} 
A variable-gadget corresponds both to the positive and negative literals associated with the variable $x_i$. For each variable $x_i$, 
since each of the literals $x_i$ and $\overline{x_i}$ may appear in each of the $m$ clauses of $\theta$ at most twice, we may need a total of $2m$  
points for both $x_i$ and $\overline{x}_i$ in the variable-gadget. We create the variable gadget as follows:
\begin{itemize}
\item It is a rectangle $r_i$ of size $(4m+3)a\times 5a$, 
\item Assuming the coordinate of the bottom left  corner of this rectangle as $(0,0)$, $8m+4$ points placed along the boundary of another rectangle $r'_i$ of size 
$(4m+1)a\times 3a$ inside $r_i$ at coordinate points
$\{((i+1)a,4a), i=1,2,\ldots, 4m\}, (4m+2)a,2a), (4m+2)a,3a), \{((i+1)a,a), i=1,2,\ldots, 4m\}, (a,2a), (a,3a)$  
  (see Figure \ref{figvarpat}(a)).  
 \item Points are labeled with $x_i$ and $\overline{x}_i$ alternately around the boundary of the rectangle $r'_i$ in clockwise order starting from the point at the location having 
 coordinate $(2a,4a)$.
\end{itemize}
 We have the following observation on variable-gadgets, which we will prove in Lemma \ref{lemvrpt}.

   \begin{figure} [t]
\centering\includegraphics[scale=0.3]{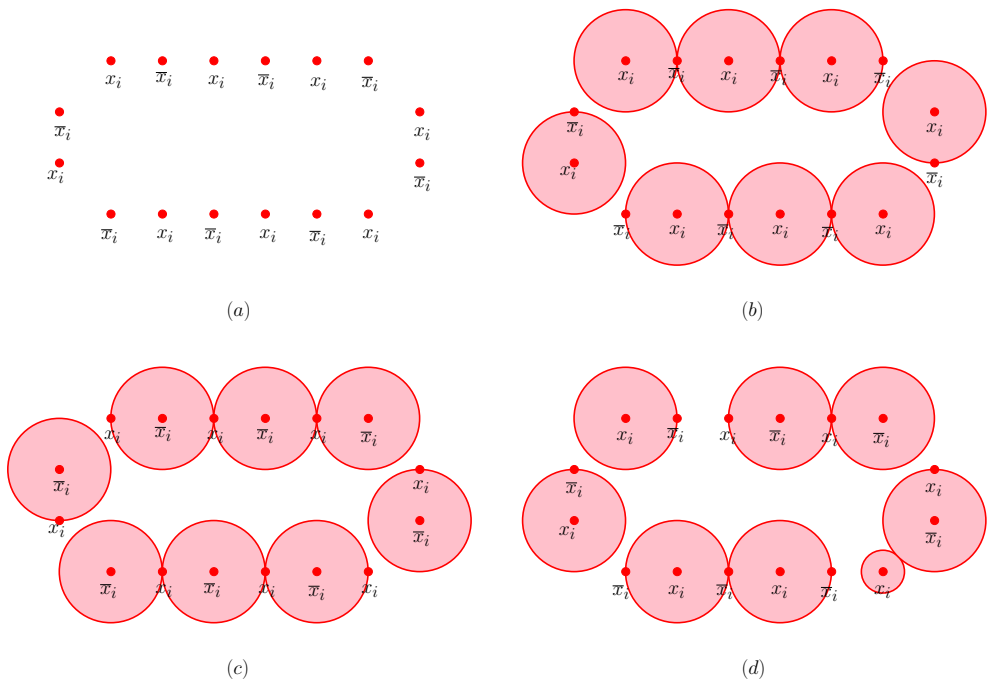}
\caption{(a) A variable-gadget with $4m$ points, (b-c) Two configurations of disks achieving maximum area, 
(d) A non-optimum configuration of disks for those points}
\label{figvarpat}
\end{figure}

\begin{observation}
\begin{itemize}
\item[(a)]  The total area of non-overlapping disks on the points of the variable-gadget for $x_i$ is maximized if and only if either all the 
  points representing $x_i$ have disks of radius $a$ on them, or all the points representing $\overline{x_i}$ 
  have disks of radius $a$ on them (see Figures \ref{figvarpat}(b) and \ref{figvarpat}(c)). In this case, the total area covered is $(4m+2)a^2$ 
\item[(b)] If disks are placed at both $x_i$ and $\overline{x}_I$ in non-overlapping manner, then the total area covered is strictly less than $(4m+2)a^2$ 
(see Figure  \ref{figvarpat}(d)).
\end{itemize}
\end{observation}

\begin{figure} 
\centering\includegraphics[scale=0.3]{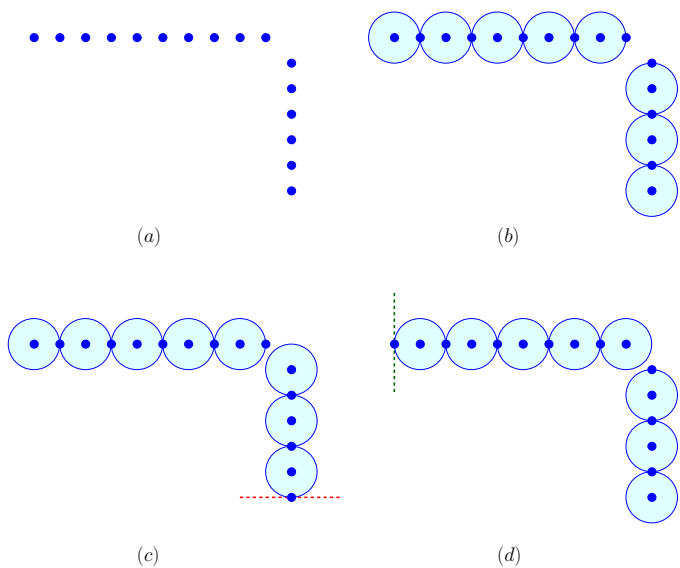}
\caption{ (a) A {\em CVC-path}. (b) A maximum disk configuration of the connecting path. 
(c) A unique maximum disk configuration of the connecting path when a red disk touches its bottommost point. 
(d) A unique disk configuration of the connecting path when a green disk touches its leftmost point. 
}
\label{figcnpath}
\end{figure}

 We construct the configuration of points $P$ for the PRM-3SAT formula $\theta$ using the following steps:
 \begin{itemize} 
  \item[(a)] Consider an embedding $\xi$ of $\theta$. The variables are represented by squares of size $(4m+1)a\times 4a$ centered on the $x$-axis, and 
  the clauses are represented by rectangles of size\footnote{length $\times$ height} $4a \times 3a$. The horizontal distance between two 
  consecutive squares on the $x$-axis (representing variables) is $2a$. The vertical distance between two rectangles defining two 
  different clauses in the embedding (if any) is also $2a$.
 \remove{ 
  \item[(b)] \label{step2}
  Stretch the embedding of $\theta$ horizontally by a factor of $2m$. The reason is that (i) a variable, say $u$, or its negation $\overline{u}$ may appear 
  at most twice\footnote{to make each clause having three literals} in a clause, (ii) $u$ and $\overline{u}$ do not occur simultaneously in a clause, and (iii)
  number of clauses in $\theta$ is $m$. Thus, each variable-square becomes a rectangle of length $4am$ after the stretch, and it contains the corresponding 
  variable-gadget as described earlier.}
  
  \item[(b)] Replace each clause-rectangle by a clause-gadget, as follows.
  \begin{itemize}
\item[]  In the original embedding $\xi$ of $\theta$, all paths from a clause to its variables are vertical lines.
  Each clause ${\cal C}=(u+v+w)$ in the embedding has three literals, namely left-literal,  middle-literal and  right-literal respectively. 
  Consider the middle-literal $v$ of a positive clause $\cal C$ embedded above the $x$-axis 
  in the embedding of $\theta$. Among all the positive clauses having literal $v$, 
  let $\cal C$ be the $k$-th one from the left in our embedding $\xi$.  
  Then place the clause-gadget $\cal C$ so that the $x$-coordinate of its
  left-most point ($p_2$ in Figure \ref{figclpt}(a))
  is greater than the $x$-coordinate of the  
  $(4k-1)^{th}$ red point in the top boundary of the variable-gadget for the variable $v$ by a multiple of $a$.
 
  If the path is from a variable to a negative clause, then follow an analogous procedure of placing the corresponding clause-gadget such that 
  the $x$-coordinate of its left-most point ($p_2$ in Figure \ref{figclpt}(a)) is greater than the $x$-coordinate of the $4k^{th}$ red 
  point in the bottom boundary of the variable-gadget of the variable $v$ (i.e., $k$-th $\overline{v}$ from the left).
\end{itemize}  
  
  \item[(c)] In the original embedding of $\theta$, all paths from clauses to variables are vertical lines.
  Consider such a vertical line $\ell$ in the embedding of $\theta$. Suppose the path connects a 
  positive clause $\cal C$ and a variable $v$. Also assume that among all such vertical paths from positive 
  clauses to the variable $v$, this path is the $k^{th}$ one from the left. Translate $\cal C$ horizontally, so 
  that it is vertically above the $(2k-1)^{th}$ red point representing the variable $v$ on the top boundary of the 
  variable gadget (rectangle) for $v$. If the path 
  corresponds to the left, middle or right literal in the clause, then add a vertical line segment of length 
  $3a$, $2a$ or $a$ respectively above it, and after that a horizontal line segment of adequate length such
  that the path is horizontally $a$ distance away from its corresponding green point of clause $\cal C$. 
  
The case when $\cal C$ is a negative clause and the vertical line $\ell$ connecting $\cal C$ and its middle-literal $\overline{v}$ is the $k$-th one 
  among the vertical lines incident on the bottom boundary of the variable gadget corresponding to $v$ in the embedding $\xi$, then 
  we translate $\cal C$ horizontally to align $\ell$ with the $2k$-th 
 red point representing $\overline{v}$ in the bottom boundary of the variable gadget of $v$. Next, we follow the same procedure (increasing the length of $\ell$ vertically 
 downwards and adding a horizontal line segment of required length) to connect $\overline{v}$ with the corresponding green point of clause $\cal C$.

\item[(d)] Now we express $a$ and $b$ in terms of $m$ and $n$. The height (span in vertical 
direction) of the embedding is upper bounded by that of $m$ clause rectangles (assuming that they are 
in different layers in the embedding $\xi$), vertical gap between layers, and a variable rectangle.  
These make a total of $3ma  + 2ma + 5a= 5(m + 1)a = B$ (say). The upper bound on the length ($L$) 
of the embedding is $(4m+5)na$. So, the length of a {\it 
CVC-path} connecting a clause with a literal is upper bounded by $K= L + B$. 
 There are at most $3m$ CVC-paths in our point set.
We want 
to set $a$ and $b$ such that the sum of the areas of all blue disks is a small fraction of the area of 
a single green or red disk. 
We want the area of a single green or blue disk to be $100$ times that of the sum of areas of all blue disks.
So, we set $a$ such that $100(3m)K \pi b^2 \leq \pi a^2$. Or, in other words,  $b \leq \frac {a} {10 \sqrt{K}} = 
\frac {a} {10 \sqrt{3m(5(m + 1)a + (4m + 5)na)}}=\frac {\sqrt a} {10 \sqrt{3m(5m + 5 + 4mn + 5n)}}$. Choosing $b=1$ gives 
$a \geq 300m {(5m + 5 + 4mn + 5n)}$. Since $5m + 5 + 4mn + 5n \geq 20mn$, we set $a = 300m { (20mn)} = 6000 m^2n$.  

\item[(e)] Note that here $\frac{a}{10}$ is an integer. Thus, the point set consisting of all the variable-gadgets and all the 
clause-gadgets can be placed at points with integer coordinates.

\item[(f)] Replace each CVC-path from a clause to a vertex with blue points at unit distance apart along that path, 
except at the turning point (see Figure \ref{figcnpath}).
The vertical and horizontal lengths of the paths are multiples of $a$, which is an even number
due to our choice of $a$. Hence, the number of lattice points on each path is odd. Since we do not put 
blue point on the turning point, the number of blue points on each path is also even. As mentioned earlier,  
the end-point of a CVC-path closer to a clause is referred to as the \emph{last blue point} of the said path
(points $u$, $v$ and $w$ shown in Figure \ref{figclpt}(a)).

\item[(g)] We use a total number of $K'$ (= $m(3 \times 100K+4)$) blue points, where $K'$ depends only on $\theta$ and not
an embedding of $\theta$.
Let the total number of blue points used so far on paths and clause-gadgets be $K_b$. Since each path and
clause-gadget has even number of blue points,
$K_b$ is even. Put $K'-K_b$ blue points on a separate vertical line,
with consecutive points at unit distance apart. These will be referred to as the {\em excess points} from the clause-gadgets. 
Note that $K'-K_b$ is also even.  
 \end{itemize}
 See Figure \ref{monsatcons} for the point set embedding $P$ of the PRM-3SAT formula shown in Figure \ref{monsat}. Here the coordinates of each point are integer.
  \begin{figure} 
\centering\includegraphics[scale=0.15]{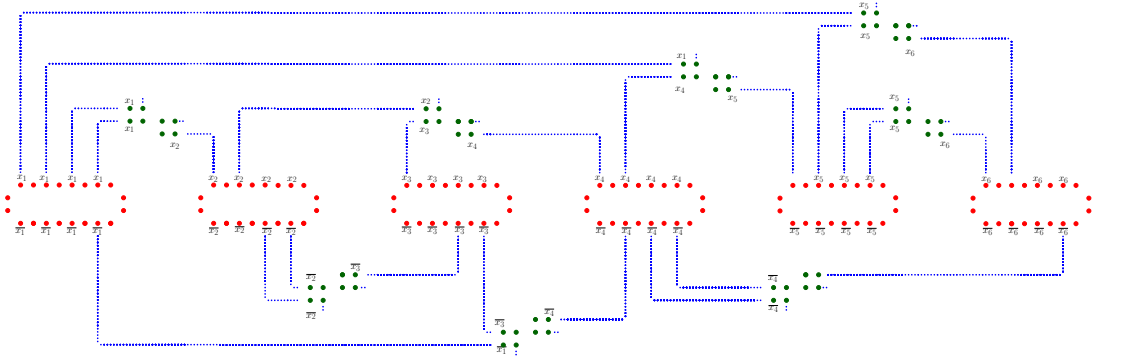}\caption{A planar monotone rectilinear 3-SAT embedding transformed to a point set.}
\label{monsatcons}
\end{figure}

 \subsection{Properties of the point configuration}
 Denote the point set constructed in the previous section by $P$.
 Denote by $\Delta (P)$ the maximum sum of areas of non-intersecting disks centered at the points of $P$.

 \begin{lemma}  \label{lemsum}
 If $Q_1$ and $Q_2$ are disjoint subsets of a point set $Q$, then $\Delta (Q) \leq \Delta (Q_1) + \Delta (Q_2)$.
 \end{lemma}
 \begin{proof}
 Assume on the contrary that $\Delta (Q) > \Delta (Q_1) + \Delta (Q_2)$ for some choice $S$ of disks, and $S=S_1 \cup S_2$, where 
 $S_1$ are centered at points in $Q_1$ and $S_2$ are centered at points in $Q_2$. Observe that the total area of $S_1$ (resp. $S_2$) 
 is smaller than $\Delta(Q_1)$ (resp. $\Delta(Q_2)$), leading to  a contradiction.
 \end{proof}
 
 \begin{lemma} \label{lemline}
 If $Q$ is a set of $k>1$ collinear points on the plane, placed uniformly unit distances apart, then $\Delta (Q) = \pi \lceil \frac{k}{2} \rceil$,
 and it can be realized only by a configurations of disks of unit and zero radii at alternate points.
 \end{lemma}
 \begin{proof}
  Follows from Lemma \ref{end_points} where the points on a line are equidistant.
 \end{proof}

\begin{figure} 
\centering\includegraphics[scale=0.25]{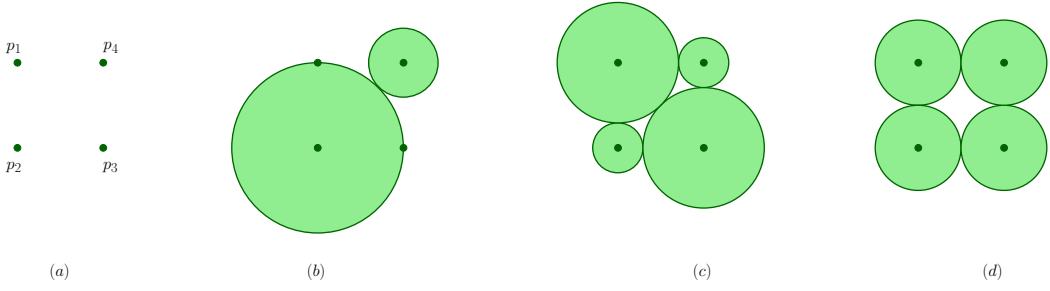}
\caption{(a) A point set with points on four vertices of a unit square. 
(b) Maximum area of disks is $\pi a^2 + \pi (\sqrt{2}-1)^2a^2 + \pi b^2$. (c) Another choice for maximum area of disks.
(d) A non-optimum area covered by disks on the points.
}
\label{figsqr}
\end{figure}
\remove{ 
 When a disk configuration can be reflected or rotated to another disk configuration on the same point set,
 then we consider the two configurations as the same. 
 We have the following lemmas.
}
 \begin{lemma} \label{lemsqr}
 If $Q$ is a set of four points on four vertices of a unit square, then $\Delta (Q) = \pi (4-2 \sqrt {2})$,
 and can be realized only by
 either two diagonally opposite disks of radius  $\frac {\sqrt 2} {2}$
and two other diagonally opposite disks of radius $1 - \frac {\sqrt 2} {2}$, or two diagonally opposite disks of
radii $1$ and $\sqrt 2 - 1$, respectively.
\end{lemma}
 \begin{proof}
 We prove this result by exhaustive case analysis.
Let the top left point of $Q$ be $q_1$, and the other points are named as $q_2$, $q_3$ and $q_4$ in a counterclockwise order
(see Figure \ref{figsqr}(a)). Let the disks on these points be named as 
 $d_1$, $d_2$, $d_3$ and $d_4$, and their 
 radii be $r_1$, $r_2$, $r_3$ and $r_4$ respectively, where $r_i \geq 0$ for $i=1,2,3,4$. 

 In a configuration achieving the maximum area, each disk $d_i$ must touch some other disk in $d_j, j\neq i$, $i,j=1,\ldots,4$. Suppose that in a configuration,  
 there are three or less disks having 
 strictly positive radii. Let the point $q_1$ has no disk (i.e., $r_1 =0$). It must be touched by some other disk, say $d_2$ at point $q_2$, having $r_2 > 0$.
Implying, $r_2=1$, and it touches another point $q_3 $ as well. The disk $d_4$ can have a radius up to $\sqrt 2 - 1$ to avoid intersection with $d_2$.
Adding the area of $d_2$ and $d_4$, we have  $\Delta(Q) \geq \pi (4-2 \sqrt {2})$ (see Figure \ref{figsqr}(b)). 

Now suppose that there are fours disks on the four points with maximum possible total area.
Suppose that one of these disks, say $d_1$, touches all other three disks. 
Then the total area $A(r_1)=\pi(r_1^2 + 2(1-r_1)^2 + (\sqrt 2 -r_1)^2) = 4\pi r_1^2-(4+2 \sqrt 2 )\pi r_1 +4 \pi$  is a convex function
of $r_1$. It attains minimum at $r_1^*=\frac{1}{2} + \frac{\sqrt 2}{4}$, and increases in both the sides of $r_1^*$.  
We also have $r_1 + r_2 = 1$, $r_1 + r_4 = 1$, and $r_1 + r_3 = \sqrt 2$, and $r_2+r_3 \leq 1$. Thus, $r_1+r_2+r_3\leq \frac{1+\sqrt{2}}{2}$. 
Implying $r_1 \geq \frac{\sqrt 2}{2}$. 
But $A(\frac{\sqrt 2}{2}) = \pi(4 - 2 \sqrt 2) = A(1)$ (see Figure \ref{figsqr}(c)).  

Now we consider the case where no disk touches all the three other three disks. So, each disk must touch either one or two other points or disks. 
Here two cases need to be considered.
\begin{description}
\item[] A pair of diagonally opposite disks, say $d_1$ and $d_3$  touch each other. 

If $r_1 > r_3$, then 
$r_2$ and $r_4$ can be set appropriately such that $d_2$ and $d_4$ touch $d_1$, leading to a contradiction. 
The same argument holds for $r_1 < r_3$.

If $r_1 = r_3$, then as before, each of $d_2$ and $d_4$
must touch both $d_1$ and $d_3$, a contradiction. 

\item[] No pair of diagonally opposite disks are touching. Let $d_1$ touch $d_2$, $d_2$ touch $d_3$, and $d_3$ touch $d_4$,  Implying $r_1+r_2=1$, $r_2+r_3=1$, and 
$r_3+r_4=1$. Thus $r_1+r_4=1$.
As we have assumed that diagonally opposite disks are non-touching, all the four disks must have radius less than $\frac{\sqrt 2}{2}$.
Let $r_1 = r$, and the total area becomes $A(r)=2 \pi (r^2 + (1 - r)^2) = 2 \pi (2r^2 - 2r + 1)$, which is an unimodal function. 
It attains minimum at $r =  \frac{1}{2}$, and monotonically increases in both the sides (see Figure \ref{figsqr}(d)). Due to our constraints, 
$r \in [1 - \frac {\sqrt 2} {2},\frac {\sqrt 2} {2}]$.
Also note that, $A(1 - \frac {\sqrt 2} {2}) = A( \frac {\sqrt 2} {2})= \pi(4 - 2 \sqrt 2)$.
\end{description}
Now consider the remaining case, where $d_1$ touches $d_2$ and $d_3$ touches $d_4$, and no other two disks touch each
other. Since for any two points, the larger 
disk can be expanded and the other one shrinks to increase their total area, the larger among $d_1$ and $d_2$ can be expanded
to touch $d_3$ or $d_4$, giving a greater 
total area, a contradiction.

We have considered all the possibilities of maximum area of disks, and these give only two configurations: either two diagonally
opposite disks are of radius  $\frac {\sqrt 2} {2}$
and the other two diagonally opposite disks are of radius $1 - \frac {\sqrt 2} {2}$, or two diagonally opposite disks of radii $1$
and $\sqrt 2 - 1$. In both the cases, we have $\Delta (Q) = \pi (4-2 \sqrt {2})$. 
 \end{proof}

\begin{lemma} \label{lemclpt}
The maximum area covered by a disk configuration of the clause-gadget is $2\pi (4-2 \sqrt {2})a^2 + 2 \pi b^2$.
\end{lemma}
\begin{proof}
We divide the points $P_C$ in the clause-gadget ${\cal C}$ into the subsets: $P_1 = \{p_1, p_2, p_3, p_4 \}$, $P_2 \{p_5, p_6, p_7, p_8\}$ (of green points), and 
$P_3 \{p_9, p_{10} \}$, $P_4 \{p_{11}, p_{12} \}$ (of blue points). By Lemma \ref{lemsqr}, $\Delta (P_1) = \Delta (P_2) = \pi (4-2 \sqrt {2})a^2$.
By Lemma \ref{lemline}, $\Delta (P_3) = \Delta (P_4) = 2 \pi b^2$.
Using Lemma \ref{lemsum}, we have $\Delta (P_{\cal C}) \leq 2\pi (4-2 \sqrt {2})a^2 + 2 \pi b^2$. The disks, if any,
on $p_2$ and $p_7$ cannot intersect the disks
on $p_{10}$ and $p_{12}$ respectively, as $a^2 + (\frac{a}{10})^2 > (a+b)^2$. Thus, $\Delta (P_{\cal C}) = 2\pi (4-2 \sqrt {2})a^2 + 2 \pi b^2$, which 
can be achieved by all five configurations in Figure \ref{figclpt}.
\end{proof}
\begin{lemma} \label{lemclpt2}
Each of the maximum disk configurations in a clause-gadget (see Figure \ref{figclpt}) must touch the blue point of at least one of the three {\em CVC-path} 
(namely $u,v,w$).
Moreover, given any one of the three such points, there is a maximum disk configuration of the clause-gadget that touches 
only that point.
\end{lemma}
\begin{proof}
From Lemma \ref{lemsqr}, we know that in the optimal disk configuration of the clause-gadget of a clause $\cal C$, 
either two diagonally opposite disks of radius  $\frac {\sqrt 2} {2} a$
and the other two diagonally opposite disks of radius $(1 - \frac {\sqrt 2} {2})a$, or two diagonally opposite disks of radii $a$ and $(\sqrt 2 - 1)a$
can be placed on each of the green point sets $P_1$ and $P_2$ of $P_{\cal C}$ to get 
the maximum area for $P_1\cup P_2$. Hence in an optimum covering for $P_{\cal C}$, these remain
the only choices for the green points.
Since the possible four radii for the green points for attaining optimality are greater than $\frac {a} {10}$, no disk can be drawn on $p_8$.
Again, due to the presence of the blue point $p_9$ at a distance $\frac{a}{2}$ from $p_3$, $p_3$ can have only a disk of radius $(1 - \frac {\sqrt 2} {2})a$. Thus, the possible 
configurations of disks at the green points of $Q$ are: 
\begin{itemize}
\item[(a)] $p_5$ and $p_7$ have disks of radii $a$ and $(\sqrt 2 - 1)a$ respectively, and $p_2$ and $p_4$ have disks of radii $a$ and $(\sqrt 2 - 1)a$ respectively: This is a valid configuration, and the disk of radius $a$ on $p_2$, touches $u$ (see Figure \ref{figclpt}(f)).
\item[(b)] $p_5$ and $p_7$ have disks of radii $(\sqrt 2 - 1)a$ and $a$ respectively, and $p_2$ and $p_4$ have disks of radii $a$ and $(\sqrt 2 - 1)a$ respectively: This is a valid configuration, and the disk of radius $a$ on $p_7$, touches $w$ (see Figure \ref{figclpt}(g)). 
\item[(c)] $p_5$ and $p_7$ have disks of radii $a$ and $(\sqrt 2 - 1)a$ respectively, and $p_2$ and $p_4$ have disks of radii $(\sqrt 2 - 1)a$ and $a$ respectively: This is an invalid configuration since the disks of radius $a$ on $p_4$ and $p_5$ intersect.
\item[(d)] $p_5$ and $p_7$ have disks of radii $(\sqrt 2 - 1)a$ and $a$ respectively, and $p_2$ and $p_4$ have disks of radii $(\sqrt 2 - 1)a$ and $a$ respectively: This is a valid configuration and the disk of radius $a$ on $p_7$, touches $w$ (see Figure \ref{figclpt}(b)).
\item[(e)] $p_5$ and $p_7$ have disks of radii $a$ and $(\sqrt 2 - 1)a$ respectively, and $p_1$ and $p_3$ have disks of radii $a$ and $(\sqrt 2 - 1)a$ respectively: This is a valid configuration and the disk of radius $a$ on $p_1$, touches $v$ (see Figure \ref{figclpt}(e)).
\item[(f)] $p_5$ and $p_7$ have disks of radii $(\sqrt 2 - 1)a$ and $a$ respectively, and $p_1$ and $p_3$ have disks of radii $a$ and $(\sqrt 2 - 1)a$ respectively: This is a valid configuration and both the disks of radius $a$ on $p_1$ and $p_7$ touch $v$ 
and $w$ respectively (see Figure \ref{figclpt}(d)).
\item[(g)] $p_5$ and $p_7$ have disks of radii $(\sqrt 2 - 1)a$ and $a$ respectively, and $p_1, p_2, p_3, p_4$ have disks of radii $ 
(1 - \frac {\sqrt 2} {2})a, \frac {\sqrt 2} {2} a, (1 - \frac {\sqrt 2} {2})a, \frac {\sqrt 2} {2}a$: This is a valid configuration where the disk of 
radius $a$ on $p_7$ touches $w$ (see Figure \ref{figclpt}(c)).
\item[(h)] $p_5$ and $p_7$ have disks of radii $a$ and $(\sqrt 2 - 1)a$ respectively, and $p_1, p_2, p_3, p_4$ have disks of radii $ 
(1 - \frac {\sqrt 2} {2})a, \frac {\sqrt 2} {2} a, (1 - \frac {\sqrt 2} {2})a, \frac {\sqrt 2} {2}a$: This is an invalid configuration since $d_4$ and $d_5$ overlap. 

\end{itemize}

For the second part of the lemma, observe that the configurations of Figures \ref{figclpt}(b), \ref{figclpt}(e) and \ref{figclpt}(f)
touch only $w$, $v$ and $u$ respectively.
\end{proof}

\begin{lemma} \label{lemvrpt}
  The total area of disks on the points of the variable-gadget for a variable 
  $x_i$ is maximized if and only if either all the 
  points representing $x_i$ have disks of radius $a$ on them, or all the points representing $\overline{x_i}$ 
  have disks of radius $a$ on them, and is equal to $2\pi m a^2$.
\end{lemma} 
\begin{proof} Follows from Lemma \ref{lemline}.
\remove{By construction, a variable-gadget is a set of collinear points, with consecutive points equal distance apart. Furthermore, they
alternately represent $x_i$ and $\overline {x_i}$. Hence, by Lemma \ref{lemline} their area can be maximized only by drawing disks of radius
$a$ either on all points representing $x_i$, or on all points representing $\overline {x_i}$.
Lemma \ref{lemline} also directly implies that the maximum area is $\pi (n^x_i + n^{\overline{x}}_i) a^2$.}
\end{proof}
\begin{figure} 
\centering\includegraphics[scale=0.15]{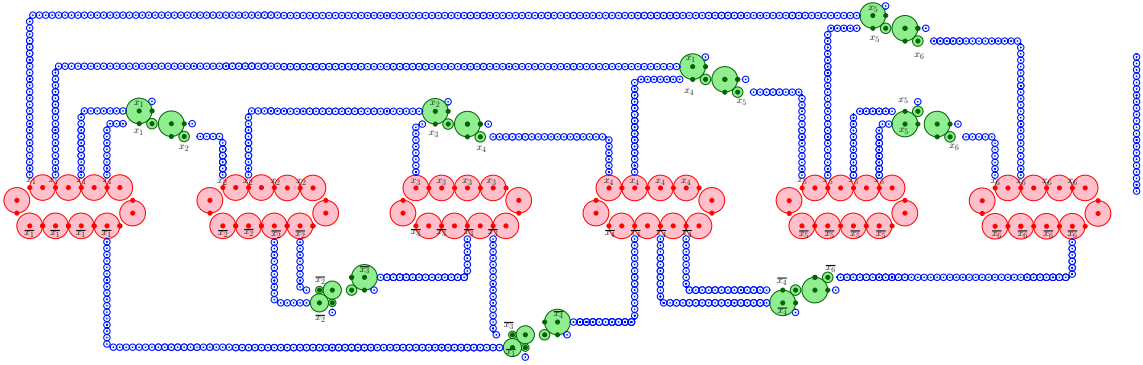}
\caption{Disks drawn according to a satisfying assignment of a PRM-3SAT formula $\theta$; here $x_1, x_2, x_5, x_6 = 1$, and $x_3, x_4 =0$}
\label{disksat}
\end{figure}

\remove{
\begin{lemma} \label{lempth1}
Every {\em CVC-path} has an even number of blue points in both its vertical and horizontal parts.
\end{lemma}
\begin{proof}
The length of both the vertical and horizontal segments of a {\em CVC-path} are multiples of $\frac{a}{10}$, which is an even number by the choice of $a$. 
Since equidistant blue points are placed at unit distance apart from one end of the path up to its other end excepting the corner joining the vertical 
and horizontal segments, the result follows.
\end{proof}
}

It is already mentioned in the earlier subsection that every {\em CVC-path} is of even length. The following lemma gives the optimum area of disks centered 
at the points on each {\em CVC-path}.
\begin{lemma} \label{lempth2}
Every {\em CVC-path}  of length $k$ has exactly three distinct maximum disk configurations, each having an area of $\pi \frac{k}{2} b^2$.
If we are not allowed to draw a disk on any one of the end points of a connecting path,
then it has exactly one possible maximum disk configuration, also having an area of $\pi \frac{k}{2} b^2$.
\end{lemma}

\begin{proof}
Consider a {\em CVC-path} $\chi$ of $k$ blue points; the set of blue points on its vertical and horizontal parts are denoted as $\chi_V$ and $\chi_H$ respectively.
By Lemma \ref{lemsum}, $\Delta {\chi} \leq \Delta {\chi_V} + \Delta {\chi_H} = \pi \frac{k1}{2} b^2 + \pi \frac{k1}{2} b^2 = \pi  \frac{k}{2} b^2$. 
Equality is attained by each of the disk configurations in Figure \ref{figcnpath}. 

Without loss of generality, let us assume that the {\em CVC-path} $\chi$  
travels vertically upward and then turns left to meet the clause-gadget.
Since the topmost blue point of $\chi_V$ and the rightmost blue point of 
$\chi_H$ are only $\sqrt 2$ distance apart,  
both of them can not have disks on them in a maximum disk configuration. As 
the number of blue points on both $\chi_V$ and $\chi_H$ are even, if there
is no disk on the bottommost blue point of $\chi_V$ in a maximum disk configuration, then there must be a disk on the topmost blue point of $\chi_V$. This implies, there is no disk 
on the rightmost blue point of $\chi_H$ and hence a disk is present on the leftmost blue point of $\chi_H$, giving an optimal configuration (see Figure \ref{figcnpath}(c)).
 Similarly, if there is no disk on the leftmost blue point in $\chi_H$, then there is a disk on the bottommost blue point of $\chi_H$. 
Thus, the other two optimal configurations are formed with  disk on  the lowest blue point in $\chi_V$ and the leftmost blue point in $\chi_H$
(see Figure \ref{figcnpath}(b)), and  disk on  
the lowest blue point in $\chi_V$ and the rightmost blue point of $\chi_H$
(see Figure \ref{figcnpath}(d)).
\end{proof}

 Now, we consider the set of points $P$ in all the clause-gadgets, 
 variable-gadgets, and CVC-paths created for a PRM-3SAT formula $\theta$. 
 The following two lemmas give  estimates of the 
 total area in the optimum solution of the MADP problem for $P$ for 
 the case where $\theta$ is satisfied, and $\theta$ is unsatisfied.
   
\begin{lemma} \label {lemoneway}
 If $\theta$ has a satisfying assignment, then there is a choice of non-intersecting disks on the points of $P$ such that 
 their total area is exactly equal to 
 $\pi ((2n+(8-4\sqrt{2}))ma^2 + (\frac{K'}{2}+2m))$.
\end{lemma}
 \begin{proof}
 Consider a satisfying assignment of $\theta$. For each variable-gadget $u$, 
 if $u=1$ then draw the disk centered at $\overline{u}$, otherwise 
  draw the disk centered at $u$. 
We also draw the disks of radius $b=1$ for the half of the extra $K'-K_b$ points.      
Thus, half of the excess blue points for each clause contains disks.  For each clause-gadget, say ${\cal C}_\alpha=x_i+x_j+x_k$, draw an
 optimum disk configuration satisfying Lemma \ref{lemclpt2}, such that one disk of radius $a$ on a green point
 must touch the {\em last blue point} of exactly one satisfying variable (literal), say $x_i$. Thus, if $x_i=1$ satisfies the clause ${\cal C}_\alpha$ 
 then the disk at the last point $u$ of the CVC-path from $x_i$ to the clause-gadget of ${\cal C}_\alpha$ can not be drawn. Now, we can put 
 $\frac{\pi k}{2}$ blue disks (of radius $b=1$) on the CVC-path 
 connecting ${\cal C}_\alpha$ to the variable gadget of $x_i$  since $x_i=1$ and so the disks at variable $x_i$ are  put on the points marked as $\overline{x_i}$.
 For the other two variables, namely $x_j$ and $x_k$ also, we can put
 exactly $\frac{\pi k}{2}$ disks on their corresponding CVC-path irrespective of whether the disks 
 are put at $x_j$ or $\overline{x_j}$ (resp. $x_k$ or $\overline{x_k}$) for the variable $x_j$ (resp. $x_k$) since the last point $v_j$ (resp. $v_k$) near to 
 the clause $\theta_\alpha$ may contain a disk. 
 Thus, the total area for all clause-gadgets is $\pi(2mb^2 +  2m(4-2\sqrt 2)a^2)$ (see Lemma \ref{lemclpt}), total area for $n$ variable-gadgets is $2 \pi mna^2$ 
 (see Lemma \ref{lemvrpt}), and total area for all CVC-paths
 is $\frac{K'}{2}\pi b^2$. See Figure \ref{disksat} for the demonstration. Thus, the total area is $\pi((8-4\sqrt{2})ma^2+2mb^2 + 3mna^2 + \frac{K'}{2}b^2)$. Putting $b=1$ and simplifying, the result follows. 
  \end{proof} 
  
\begin{lemma} \label{lemoppway}
If a PRM-3SAT formula $\theta$ is not satisfiable then the total area of the corresponding MADP problem is less than $\pi ((2n+(8-4\sqrt{2}))ma^2 + (\frac{K'}{2}+2m))$.
\end{lemma}
 \begin{proof}
 Let $\theta$ is not satisfiable. There is no difficulty to have a total 
 area of $3\pi mn a^2$ from the variable-gadgets corresponding to $n$ variables since the used disks at 
 the red points of the variable gadgets are much larger than the disks used for the blue points of the CVC-paths near them. 
 Similarly, the disks used for the green points are much larger 
 than the disks used for the blue points in it and also the last point on 
 its adjacent three CVC-paths. Moreover, among the four blue points $p_9$, $p_{10}$, $p_{11}$ and $p_{12}$, 
 $p_{10}$ and $p_{12}$ can always be used for placing disks of radius $b=1$. Also, exactly two such disks can be placed irrespective of any arbitrary assignment of disks 
 among the other 8 green points in that clause-gadget. Thus, for each 
 clause-gadget exactly  $\pi(2m +  2m(4-2\sqrt 2)a^2)$ area is achieved 
 in the optimum solution in the MADP problem with the point set $P$ 
 corresponding to $\theta$. Now, let us consider the CVC-paths. As $\theta$ 
 is not satisfiable, for each truth assignment $X$ of the variables there is at least one clause, say $\theta_\alpha$, that is not satisfiable. In the 
 optimum disk assignment of the green points of $\theta_\alpha$, at least 
 one of the green disks must touch the last (blue) point of the corresponding  CVC-path. For the CVC-path, that is not touched by any green disk, one can 
put a disk (of radius $b=1$) at its last point, and a total area of 
$\pi\frac{k}{2} $ is achievable on the blue points along that path even if 
 disk can not be placed at the other end of that CVC-path. But, if the last vertex of a CVC-path is touched by a green disk, no blue disk 
 can be placed at its either end. Thus only a total area of  $\pi(\frac{k}{2}-1)$
 is achievable. Thus, for the truth assignment $X$, the total area obtained 
 in the optimum solution of the MADP problem is at most $\pi ((11-4\sqrt{2})ma^2 + (\frac{K'}{2}+2m)-\beta)$, where $\beta$ is the number of non-satisfied clause(s).
 The result follows from the fact that if $\theta$ is not satisfiable, $\beta \geq 1$ for every truth assignment of the variables.   
   \end{proof}
Thus, we can check the satisfiability of a PRM-3SAT formula with $m$ clauses and $n$ variables by generating 
the points $P$ as described, and then observing whether the total 
area of the disks in the optimum solution of the MADP 
problem on the point set $P$ is equal to or less than $\pi ((2n+(8-4\sqrt{2}))ma^2 + (\frac{K'}{2}+2m))$. As PRM-3SAT problem is NP-complete \cite{plane-rec}, we have the following result.

\begin{theorem}
 The problem of finding disks of a maximum total area centered on a given set of points, is NP-hard.
\end{theorem}

\subsection{MADP for axis-parallel squares}
 Now, we demonstrate that the MADP problem remains NP-hard when the objects are 
 axis-parallel squares instead of disks.
  
Our reduction, as before, follows from PRM-3SAT. In fact, we just modify the point configuration for disks to get the reduction for squares.
Our clause patterns are now simplified, with only six points, shown in Figure \ref{squareclause}(a). Only three maximum configurations are possible,
shown in Figures \ref{squareclause}(b),  \ref{squareclause}(c) and \ref{squareclause}(d).
The variable patterns remain identical with only two possible maximum configurations, as shown in Figures \ref{squarepath1}(a) and  \ref{squarepath1}(b).
The CVC-paths also  remain identical with only two possible maximum configurations  having no square around one of the end points, as shown in Figures \ref{squarepath2}(a), \ref{squarepath2}(b),
\ref{squarepath2}(c) and \ref{squarepath2}(d). The reduction proceeds as before, with the squares giving a maximum area of
$4(4ma^2+ +2mna^2+50K+m)$ if and only if the corresponding PRM-3SAT formula is satisfiable.

\begin{figure}[h] 
\centering\includegraphics[scale=0.2]{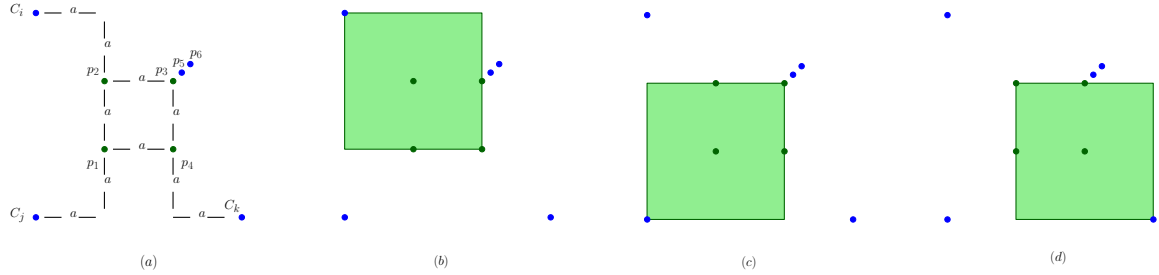}
\vspace{-0.1in}
\caption{Clause gadget: MADP problem for squares}
\label{squareclause}
\end{figure}

\begin{figure}[h] 
\centering\includegraphics[scale=0.12]{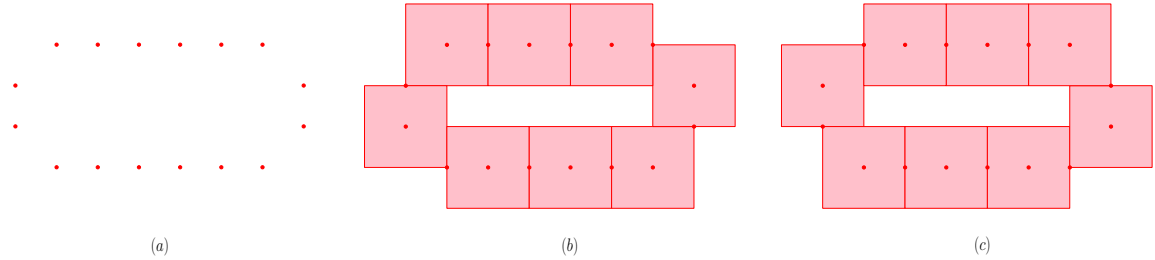}
\vspace{-0.1in}
\caption{Variable gadget: MADP problem for squares}
\label{squarepath1}
\end{figure}

\begin{figure}[h] 
\centering\includegraphics[scale=0.18]{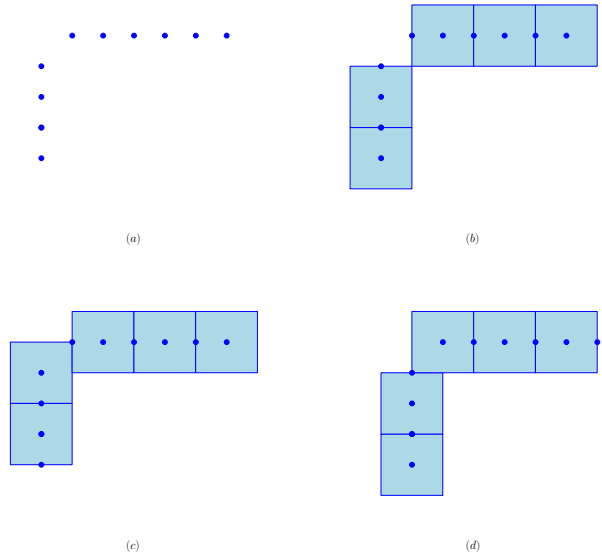}
\vspace{-0.1in}
\caption{CVC-path gadget: MADP problem for squares}
\label{squarepath2}
\end{figure}

\section{Approximation algorithm}\label{section4}
In this section, we show that the optimum solution for the MPDP problem proposed in \cite{eppstein}
gives a 2-factor approximation result for the MADP problem. We also propose a PTAS for the problem.

\subsection{$2$-factor approximation algorithm}\label{1D_Approx}
Given a set of points $P$ in $\IR^2$, let 
$R=\{r_i, i=1,2,\ldots,n\}$ be the set of radii of 
the points in $P$ obtained by the optimum solution for  MPDP  problem \cite{eppstein}.
It is clear that any feasible solution of the MPDP is  a feasible 
solution of the MADP problem. We show that an optimal radii returned by the MPDP problem produce at most $2\times 
OPT$ area for the corresponding MADP problem, where $OPT$ is the optimum solution of that MADP problem. 

\begin{lemma} \label{x}
\cite{eppstein} The maximum sum of radii of non-overlapping disks, centered at points 
$p_i \in P$, equals half of the minimum total edge length of a collection of vertex-disjoint 
cycles (allowing 2-cycles) spanning the complete geometric graph on the points $p_i\in P$ 
with each edge having length equal to the distance between the end-points of that edge. 
\end{lemma}

\begin{lemma} \label{y}
\cite{eppstein} In the minimum total edge length of a collection of vertex-disjoint 
cycles mentioned in Lemma \ref{x}, each cycle is either of odd length or a  2-cycles (i.e., a single edge).
\end{lemma}

The implication of Lemma \ref{x} and \ref{y} is that {\em in the optimum solution of the 
MPDP problem, each disk touches its neighboring disk(s) in the cycle in which it appears.}

In \cite{eppstein}, an $O(n^{1.5})$ time algorithm is proposed to compute the minimum length 
cycle cover $\cal C$ of the complete geometric graph $G$ with a set $P$ of $n$ points on 
the plane. From the geometric property of the Euclidean distances, they show that if a 
subgraph $G'$ of $G$ is formed by removing all the edges $(p_i,p_j)$ satisfying 
$\text{dist}(p_i,{\cal N}(p_i)) + \text{dist}(p_j,{\cal N}(p_j)) < \text{dist}(p_i,p_j)$, 
then the minimum weight  cycle cover of $G'$ remains same as that in $G$. 

\begin{lemma}\label{lx}
For a given set of points $P$ arbitrarily placed in $\IR^2$, 
the radii $\{r_i, i=1,2,\ldots,n\}$ in the optimum solution of the MPDP problem is a $2$-approximation 
result for the MADP problem for the point set $P$. 
\end{lemma}

\begin{proof}
As mentioned, MPDP algorithm generates the cycles ${\cal C}=\{C_1, C_2, 
\ldots, C_k\}$.  We need to show that $\sum_{\alpha=1}^n r_\alpha^2 \geq 
\frac{1}{2}\sum_{\alpha=1}^n \rho_\alpha^2$, where $\rho_\alpha$ is the radius in the optimum solution of the MADP problem for the point $p_\alpha$.
We show that $\sum_{p_\alpha \in C_i}^n r_\alpha^2 \geq 
\frac{1}{2}\sum_{p_\alpha\in C_i}^n \rho_\alpha^2$ for each cycle $C_i\in {\cal C}$.   As each disk participates in exactly one
of the cycles,  agreegating these relations for all the cycles $C_i, i=1,2,\ldots, k$, we will have the desired result.
Let us consider the following two cases separately.
\begin{description}
\item[$\mathbf{C_i}$ is a 2-cycle $\mathbf{(p_\alpha,p_\beta)}$:] Let $r=\text{dist}(p_\alpha,p_\beta)$. As the disks
centered at  $p_\alpha$ and $p_\beta$ are touching each other, let $r_{\alpha}=\frac{r}{2}-\delta$ and  
$r_{\beta}=\frac{r}{2}+\delta$.  Thus, $r_{\alpha}^2+r_{\beta}^2  \geq \frac{r^2}{2}$.

 Note that in the optimum solution of the 
 MADP problem, the disks for $p_\alpha, p_\beta$ may not be touching, but $\rho_\alpha +\rho_\beta \leq \text{dist}(p_\alpha,p_\beta)$.
 So, the upper bound of the sum of squares of the radii in the optimum solution is:
 $\rho_\alpha^2 +\rho_\beta^2 \leq (\rho_\alpha +\rho_\beta)^2  \leq  (\text{dist}(p_\alpha,p_\beta))^2 = r^2$.

 Thus, for the two-cycle  ${\cal C}_i=
(p_\alpha,p_\beta)$, we have $r_\alpha^2+r_\beta^2 \geq \frac{1}{2}(\rho_\alpha^2 +\rho_\beta^2)$.

\item[$\mathbf{C_i}$ is an odd cycle:] Let the length of the cycle 
be $m$. Without loss of generality, assume that the vertices be $p_1, p_2, \ldots, p_m$.
For each edge $(p_\alpha,p_{\alpha+1})$ of this cycle (where 
the indices are numbered modulo $m$), we have $r_\alpha^2+r_{\alpha+1}^2 \geq 
\frac{1}{2}(\rho_\alpha^2 +\rho_{\alpha+1}^2)$ (as explained in the 
earlier case). Adding these inequalities for $\alpha=1,2,\ldots,m$, 
we have $2\sum_{\alpha=1}^m r_\alpha^2 \geq \frac{1}{2}[2\sum_{\alpha=1}^m \rho_\alpha^2]$.
Ignoring the factor 2 in both sides, we have the result.
\end{description}\vspace{-0.2in}
\end{proof}

Combining Lemma \ref{lx} with the time complexity result in \cite{eppstein}, we have the following result.
\begin{theorem}
 For a given set of points $P$ arbitrarily placed in $\IR^2$, one can compute a 2-approximaton result of
 the MADP problem in $O(n^{\frac{3}{2}})$ time.
\end{theorem}

\subsection{PTAS} \label{ptas}
In this section, we propose a PTAS for the MADP problem. In \cite{Erlebach}, Erlebach et al. proposed a $(1+\frac{1}{k})$-factor 
approximation algorithm for the maximum weight independent set  for the intersection graph of a set of weighted disks of arbitrary 
size. It runs in $n^{O(k^2)}$ time. We will use this algorithm in designing our PTAS. 

For each point $p_i\in P$,  let the maximum possible radius be   $\ell_i=\text{dist}(p_i,{\cal N}(p_i))$. Thus, the maximum possible area be 
$\alpha_i = \pi \ell_i^2$. Given an integer $k$, we compute $h_i=\frac{\alpha_i}{k}$, and define $k+1$ circles ${\cal C}_i=\{C_0^i, 
C_1^i, \ldots, C_k^i\}$ centered at $p_i$ with area $\{0, h_i, 2h_i, \ldots, kh_i\}$ (see Figure \ref{ptas1}). Each disk is assigned weight equal to its 
area. Now we consider all the disks $\cup_{i=1}^n {\cal C}_i$, and use the algorithm of \cite{Erlebach} to compute the maximum 
weight independent set (MWIS) $\cal A$. Note that 
the number of disks centered at any point $p_i$ present in both the optimum solution and in our algorithm for the MWIS problem 
of $\cup_{i=1}^n {\cal C}_i$ is exactly one.

\begin{figure}[h]

\centering\includegraphics[scale=0.55]{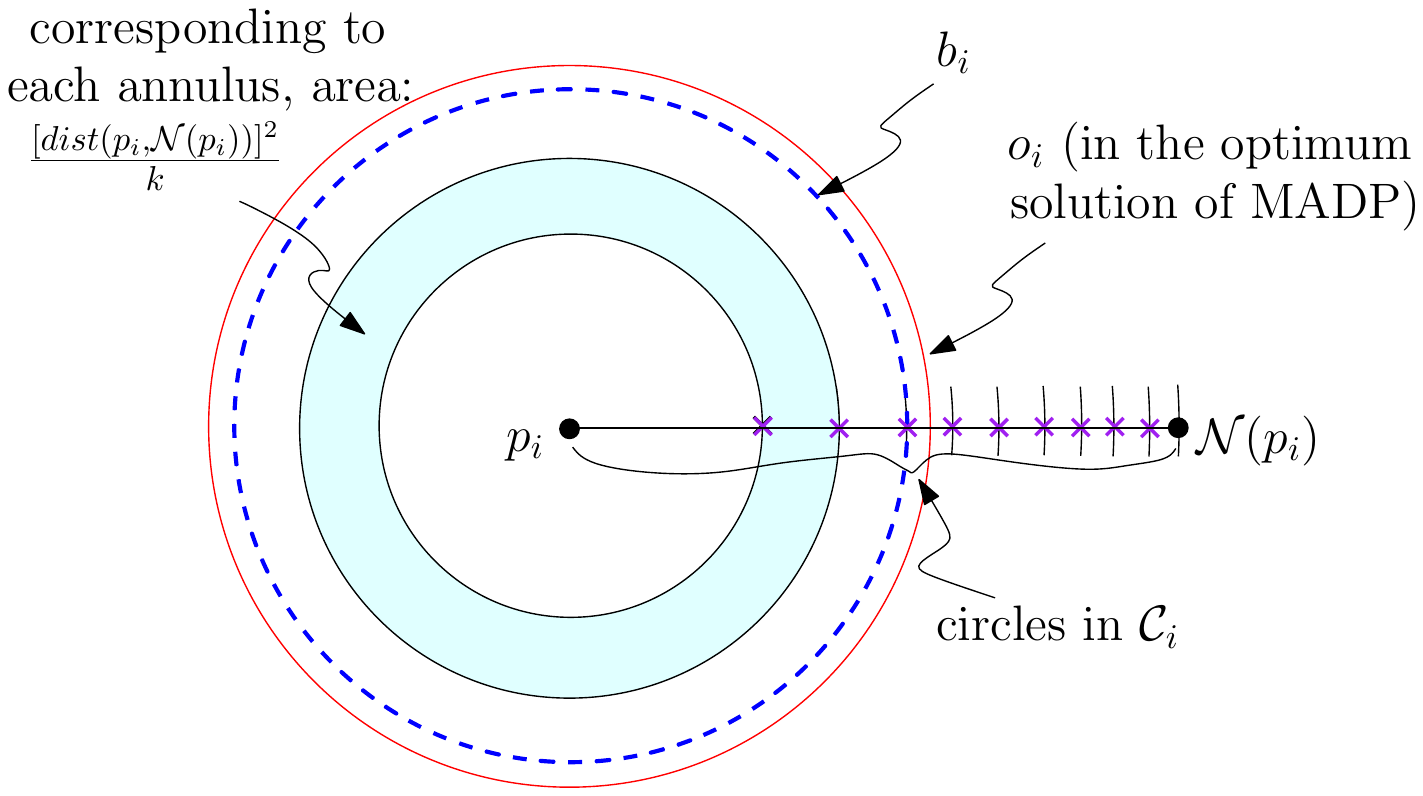}
\caption{Demonstration of PTAS}
\label{ptas1}
\end{figure}

Let $o_i$ and $a_i$ be the disks centered at $p_i$ in the optimum solution and in our solution ($\cal A$) respectively, and $O_i$, $A_i$ be their 
respective area.
Let $\Theta = \sum_{i=1}^n A_i$ be the solution obtained by our algorithm, and $OPT={\sum_{i=1}^n O_i}$ be the value of the optimum solution. 
We need to analyze the bound on  $\frac{OPT}{\Theta}$.

Let $\widetilde{OPT}$ be the optimum solution of the MWIS problem among the set of disks  $\cup_{i=1}^n {\cal C}_i$. 
Thus, $\frac{OPT}{\Theta} = \frac{OPT}{\widetilde{OPT}}\times \frac{\widetilde{OPT}}{\Theta}$. Following \cite{Erlebach}, $\frac{\widetilde{OPT}}{\Theta} 
\leq 1+\frac{1}{k}$. It remains to analyze $ \frac{OPT}{\widetilde{OPT}}$.

Now, let us consider the disks in $OPT$. For each point $p_i$, let $b_i$ be the largest disk in ${\cal C}_i$ among those which
are smaller than equal to $o_i$ (see the blue and red disks in Figure \ref{ptas1}). Thus, 
$\{b_1,b_2\ldots, b_n\}$ is a feasible solution. Let $LB(OPT)= \sum_{i=1}^n B_i$ be the lower bound of $OPT$, where $B_i$ = area of the disk $b_i$.

Since $\widetilde{OPT}$ is the optimum solution among the 
disks $\cup_{i=1}^n {\cal C}_i$, and $LB(OPT)$ is a feasible solution of the MWIS problem among the disks 
$\cup_{i=1}^n {\cal C}_i$, we have $\widetilde{OPT} \geq LB(OPT)$.

Now, consider $OPT - \widetilde{OPT} \leq OPT - LB(OPT)$ = $\sum_{i=1}^n (O_i-B_i) \leq \frac{1}{k}\sum_{i=1}^n \ell_i^2$, since $O_i-B_i \leq \frac{1}{k}\ell_i^2$ by our construction
(see Figure \ref{ptas1}).
We also have $OPT \geq \frac{1}{4}\sum_{i=1}^n \ell_i^2$ from the method of getting the 4-approximation result, mentioned in Section 1.

Thus, $\frac{OPT - \widetilde{OPT}}{OPT} \leq \frac{4}{k}$, implying $\frac{\widetilde{OPT}}{OPT} \geq 1-\frac{4}{k}$.

In other words, $\frac{OPT}{\widetilde{OPT}} \leq 1+\frac{1}{k'}$, where $k'= \frac{k-4}{4}$. Thus, $\frac{OPT}{\Theta} \leq (1+\frac{1}{k})(1+\frac{1}{k'}) \leq 
(1+\frac{1}{k''})$, where $k'' = \frac{k-4}{5}$. 
Thus, we have the following result.

\begin{theorem}
Given a set of points $P$ in ${\IR}^2$ and a positive integer $k$, we can get a $(1+\frac{1}{k})$-approximation algorithm with time complexity
$(nk)^{O(k^2)}$.
\end{theorem}

\section{MADP in higher dimension} \label{section6}
In this section, we first propose a  constant factor approximation 
algorithm for the MADP problem where the points $P=\{p_1, p_2, \ldots, p_n\}$ are distributed in ${\IR}^d$, and then we propose a PTAS for the same. 

\subsection{Approximation algorithm}

Let 
$R=\{r_i, i=1,2,\ldots,n\}$ and ${\cal R}=\{\rho_i, i=1,2,\ldots,n\}$ be the set of radii of 
the points in $P$ for the solution given by MPDP~\cite{eppstein}  and the optimum of MADP problem, respectively. As noted in Section~\ref{section2}, $R$ is a feasible solution for MADP problem, and the disks in $R$ can be partitioned into cycles where each disk is touching with two adjacent disks in the cycle.
\begin{lemma}
For a given set of points $P$ in ${\IR}^d$, 
the radii $\{r_i, i=1,2,\ldots,n\}$ is a $2^{d-1}$-factor approximation 
result for the MADP problem. 
\end{lemma}

\begin{proof}
Consider the cycles ${\cal C}=\{C_1, C_2, 
\ldots, C_k\}$ generated by  MPDP algorithm as in the proof of Lemma \ref{lx}.  
We prove that, for each cycle $C_i$, we have $\sum_{p_\alpha \in C_i}^n r_\alpha^d \geq 
\frac{1}{2^{d-1}}\sum_{p_\alpha\in C_i}^n \rho_\alpha^d$. 
Agreegating these relations for all the cycles $C_i, i=1,2,\ldots, k$, we will have the desired result. As in Lemma \ref{lx}, consider the following two cases.

\begin{description}
\item[$\mathbf{C_i}$ is a 2-cycle $\mathbf{(p_\alpha,p_\beta)}$:] Let $r=\text{dist}(p_\alpha,p_\beta)$; 
$r_\alpha=\frac{r}{2}-\delta$ and  
$r_\beta=\frac{r}{2}+\delta$ be the  radii for the two  
disks $C_\alpha', C_\beta'$  that maximizes the 
sum of square of the radii of these 
two disks, where  
$-\frac{r}{2} \leq \delta\leq \frac{r}{2}$. 
Thus,  $r_\alpha^d+ r_\beta^d=(\frac{r}{2}-\delta)^d+ (\frac{r}{2}+\delta)^d  \geq 2((\frac{r}{2})^d+{d\choose 2}(\frac{r}{2})^{d-2}\delta^2 + \ldots) \geq \frac{1}{2^{d-1}}r^d$.
The upper bound 
 for any  two non-overlapping disks having their centers 
$r$ distance apart is $\rho_\alpha^d +\rho_\beta^d \leq  (\rho_\alpha+\rho_\beta)^d \leq r^d$.
 Thus, for the two-cycle  ${\cal C}_i=
(p_\alpha,p_\beta)$, we have $r_\alpha^d+r_\beta^d \geq \frac{1}{2^{d-1}}r^d \geq \frac{1}{2^{d-1}}(\rho_\alpha^d +\rho_\beta^d)$.

\item[$\mathbf{C_i}$ is an odd cycle:] Let the length of the cycle 
be $m$. Without loss of generality, assume that the vertices be $p_1, p_2, \ldots, p_m$.
For each edge $(p_\alpha,p_{\alpha+1})$ of this cycle (where 
the indices are numbered modulo $m$), we have $r_\alpha^d+r_{\alpha+1}^d \geq 
\frac{1}{2^{d-1}}(\rho_\alpha^d +\rho_{\alpha+1}^d)$ (as explained in the 
earlier case). Adding these inequalities for $\alpha=1,2,\ldots,m$, 
we have $2\sum_{\alpha=1}^m r_\alpha^d \geq \frac{1}{2^{d-1}}[2\sum_{\alpha=1}^m \rho_\alpha^d]$.
Ignoring 2 in both sides, we have the result.
\end{description}

\end{proof}
As the time complexity of solving MPDP problem in ${\IR}^d$ is  $O(n^{2-\frac{1}{d}})$, we have the following result.
\begin{theorem}
 For a given set of points $P$ arbitrarily placed in ${\IR}^d$, one can compute a $2^{d-1}$-approximaton result of the MADP problem in $O(n^{2-\frac{1}{d}})$ time.
\end{theorem}

\subsection{PTAS}
The same scheme of designing PTAS as in Section \ref{ptas} also works in 
higher dimension due to the following reasons:
\begin{itemize}
\item The algorithm $(1+\frac{1}{k})$-factor for the maximum weighted independent set in a disk graph with geometric layout of the disks and for a given $k$ also works in 
higher dimension in $n^{O(k^{2d-2})}$ time~\cite{Erlebach}. 
\item $OPT \geq \frac{1}{2^d}\sum_{i=1}^n \ell_i^d$ in ${\IR}^d$, where $\ell_i=\text{dist}(p_i, {\cal N}(p_i))$ using the same argument as in Section \ref{ptas}, 
since the volume of a ball with radius $\frac{\ell_i}{2}$ in ${\IR}^d$ is proportionate to $\frac{\ell_i^d}{2^d}$. Thus, $\frac{OPT}{\Theta}\leq (1+\frac{1}{k})(1+\frac{2^d}{k-2^d})=(1+\frac{1}{k''})$, 
where $k''= \frac{k-2^d}{2^d+1}$. 
\end{itemize}   
 
\begin{theorem}
Given a set of points $P$ in ${\IR}^d$ and a positive integer $k$, we can get a $(1+\frac{1}{k})$-approximation algorithm in time 
$(nk)^{O(k^{2d-2})}$.
\end{theorem}

\section{MADP problem for $2m$-regular convex polygons} \label{section7}

In this section, we will show that   both 2-approximation and PTAS  results explained in Section~\ref{1D_Approx} and \ref{ptas}, respectively,  can be generalized for the MADP problem when the objective is to place non-overlapping  $2m$-regular convex polygons ($m \geq 2$) of fixed orientation centered at the given set of points  $P$ such that the sum of area  covered by them is maximized.  

For a  $2m$ regular convex polygon $S$, the \blue{\it width} of $S$, denoted by $w(S)$, 
is  defined  as 
the radius of the circle inscribed  in  $S$ which touches all the edges of 
boundary of $S$. Note that the area of $S$ is $2m {(w(S))}^2 \tan \frac{\pi}{2m}$~\cite{wiki}.

Now, we define the distance \blue{$\delta(p_1,p_2)$} between two points $p_1,p_2\in 
\IR^2$ 
as the width of the minimum width $2m$-regular convex polygon centered at the point 
$p_1$ 
containing the point $p_2$  (see Figure \ref{fig:reg_poly}).  

\begin{lemma}\label{p1}
  The distance function is symmetric\footnote{This property does not hold for odd regular convex polygon.}, i.e., $\delta(p_1,p_2)=\delta(p_2,p_1)$, 
where $p_1$ and $p_2$ are two points in $\IR^2$.
\end{lemma}

\begin{proof}
 Let $S_1$ be the 
$2m$-regular polygon centered at $p_1$ and containing the point $p_2$ on its 
boundary. Note that  $2m$-regular polygon is symmetric. 
Thus,  if we translate the polygon $S_1$ such that the center moves to $p_2$, then it 
will 
also contain the point $p_1$ on its boundary. Let this translated copy be 
$S_2$. According to the definition, $\delta(p_1,p_2)=w(S_1)$ and 
$\delta(p_2,p_1)=w(S_2)$. As $w(S_1)=w(S_2)$, hence the property follows. 
\end{proof}

\begin{figure} 
\centering\includegraphics[scale=0.4,page=2]{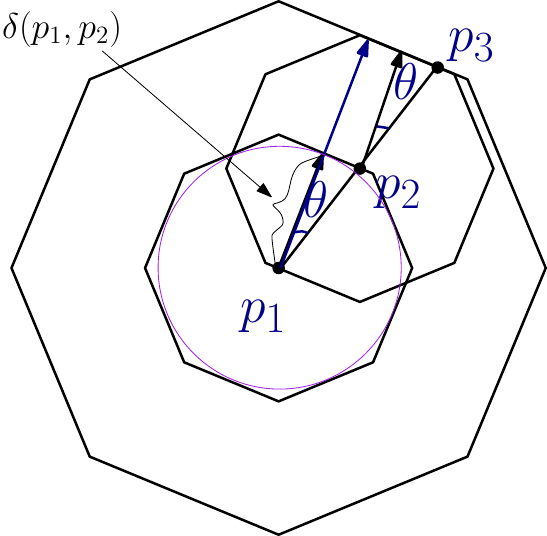}
\caption{Property of the distance function $\delta$}
\label{fig:reg_poly}
\end{figure}

\begin{lemma}\label{p3}
Let  $p_1$ and $p_3$ be any two points in $\IR^2$ and let $p_2$ be any point 
on the line segment $\overline{p_1,p_3}$, then   $\delta(p_1,p_3)= \delta(p_1,p_2)+\delta(p_2,p_3)$ .
\end{lemma}

\begin{proof}
  Let $S_1$, $S_2$ and 
$S_3$ be three $2m$-regular polygons centered at $p_1$, $p_2$ and $p_1$, 
respectively. Their widths are $\delta(p_1,p_2)$, $\delta(p_2,p_3)$ and 
$\delta(p_1,p_3)$, respectively (see Figure~\ref{fig:reg_poly}). Without loss of generality, assume that the line 
segment $\overline{p_1,p_3}$ intersects the $i$-th side of these polygons.   Let $d_j$ be the   perpendicular   
from the center  to 
the $i$-th 
side of the polygon $S_j$, for $j\in{1,2,3}$, and  $\overline{p_1,p_3}$
makes an angle $\theta$ with the perpendicular $d_j$ ($j\in \{1,2,3\}$).  
Note 
that $w(S_1)=dist(p_1,p_2) \cos \theta $, 
$w(S_2)=dist(p_2,p_3) \cos \theta $ and $w(S_3)=dist(p_1,p_3) \cos 
\theta $, where $dist(p_i,p_j)$ is the 
Euclidean distance between two points $p_i$ and $p_j$. As $p_1$, $p_2$ and 
$p_3$ 
are co-linear, so 
$dist(p_1,p_3)=dist(p_1,p_2)+dist(p_2,p_3)$.
\begin{IEEEeqnarray*}{rcl's}
dist(p_1,p_3)\cos\theta &=&  dist(p_1,p_2)\cos\theta  
+dist(p_2,p_3)\cos\theta &  \\
\Rightarrow w(P_3) & = &  w(P_1) + w(P_2) & \\
\Rightarrow \delta(p_1,p_3) &= &\delta(p_1,p_2)+\delta(p_2,p_3) &
\end{IEEEeqnarray*}
Hence the lemma follows.
\end{proof}

\begin{lemma}\label{p2}
 The distance function follows 
the triangular inequality, i.e., $\delta(p_1,p_3)\leq 
\delta(p_1,p_2)+\delta(p_2,p_3)$, where  $p_1$, $p_2$ and 
$p_3$  are any three points in $\IR^2$.
\end{lemma}

\begin{proof}
Let $S_2$ and $S_3$ be two $2m$-regular convex polygons centered at $p_1$ and containing the points $p_2$  and $p_3$, respectively, in their boundary. 
Now, if the width of  $S_3$ is less than equal to the width of $S_2$, then the lemma holds true. So, without loss of generality, assume that $\delta(p_1,p_3)> \delta(p_1,p_2)$.
Let $p$ be the intersection point of $S_2$ with the line segment $\overline{p_1p_3}$, and $S$ be the smallest $2m$-regular polygon centered at $p$   containing the point $p_3$. 
The width of  $S$ is $\delta(p_1,p_3)- \delta(p_1,p_2)$ (follows from Lemma~\ref{p3}).  If $p_2$ does not coincide with $p$, then  the translated copy of $S$ does not cover the point $p_3$ (see Figure~\ref{fig:reg_poly}). As a result, in this case, we need a $2m$-regular polygon $S'$ centered at $p_2$ of width at least $\delta(p_1,p_3)- \delta(p_1,p_2)$  to have the point  $p_3$ on its boundary.  Thus, the claim follows.

\end{proof}

\begin{lemma} 
The distance function $\delta$ satisfies the metric properties.
\end{lemma}
\begin{proof}
From the definition of the distance function, it is obvious that  $\delta(p_1,p_2)=0$ if and only if $p_1=p_2$.
Thus, the proof follows from Lemmata~\ref{p1} and  \ref{p2}. 

\end{proof}

 Combining Lemma~\ref{p1} and \ref{p3}, we have the following lemma.

\begin{lemma}
Let  $p_1$ and $p_3$ be any two points in $\IR^2$ and let $p_2$ be any point 
on the line segment $\overline{p_1,p_3}$, then   $\delta(p_1,p_3)= \delta(p_3,p_1)=
\delta(p_1,p_2)+\delta(p_3,p_2)$.
\end{lemma}

Above lemma along with the fact that  the area of a $2m$-regular convex polygon with width $w$ is $2m {w}^2 \tan \frac{\pi}{2m}$
implies that the MADP problem for $2m$-regular polygon can be formulated as a quadratic programming problem as follows. 
\begin{siderules}
\begin{tabbing}
\= Subject to \= \kill
\> Maximize \> $\sum_{i=1}^n w_i^2$ \\
\> Subject to\> $w_i+w_j \leq \delta(p_i,p_j)$, $\forall$ $p_i,p_j \in P$, $i \neq j$. 
\end{tabbing}
\end{siderules}

Similarly, we can formulate the linear programming problem of MPDP problem where the objective is to
maximize  $\sum_{i=1}^n w_i$ with the same set of constraints.
Note that Eppstein's result~\cite{eppstein} for MPDP problem holds  when distance function is  metric.
The only difference is in time complexity which takes $O(n^3)$.  Throughout our proof in
Section~\ref{1D_Approx} and \ref{ptas}, we have not used any special property of disk other than
the metric property of the Euclidean distance function. Thus,  using the distance function $\delta$
instead of Euclidean distance and assuming that $\delta(p_i,p_j)$ can be computed in constant time
(which depends only on $m$ which is constant),  we have the following results.

\begin{theorem}
 For a given set of points $P$ arbitrarily placed in the plane, one can compute a 2-approximaton
 result of the MADP problem for $2m$-regular convex polygons  in $O(n^3)$ time.
\end{theorem}

\begin{theorem}
Given a set of points $P$ in $R^2$ and a fixed integer $m$, where we have to place non-overlapping $2m$-regular convex 
polygons to maximize the area covered by them, we can get a $(1+\frac{1}{k})$-approximation
algorithm in time $(nk)^{O(k^2)}$.
\end{theorem}

Furthermore, note that the approach given in Section~\ref{section6} also works for $2m$-regular convex
polygons in fixed dimension $d$ with same approximation guarantee.

\section{Conclusion}
Following Eppstein's work \cite{eppstein} on placing non-overlapping disks for a set of given points on the plane to maximize 
perimeter, we study the area maximization problem under the same setup. If the points are placed on a straight line, then the 
 area maximization problem is solvable in polynomial time. Though the perimeter maximization problem in ${\IR}^2$ is polynomially
 solvable, the area maximization problem  is shown to be NP-hard. We also observe that the solution 
of the perimeter maximization problem gives a $2^{d-1}$-factor approximation result of the area maximization problem in ${\IR}^d$. 
A PTAS for the MADP problem in ${\IR}^d$ is also proposed. Finally, we show that
these results  for MADP problem can be generalized   for different types of objects: squares, and regular convex polygons with even number of edges.

 \bibliographystyle{plain}

\begin{thebibliography}{1}

\bibitem{BentzCR13}
C{\'{e}}dric Bentz, Denis Cornaz, and Bernard Ries.
\newblock Packing and covering with linear programming: {A} survey.
\newblock {\em European Journal of Operational Research}, 227(3):409--422,
  2013.


\bibitem{plane-rec}
Mark de Berg and Amirali Khosravi, \newblock Optimal binary space partitions in the plane. \newblock Lecture Notes in Computer Science, vol. 6196,   
        pages 216--225, 2010.


\bibitem{chang}
Hai-Chau Chang and Lih-Chung Wang.
\newblock A simple proof of Thue's Theorem on circle packing.
\newblock {\em arXiv preprint arXiv:1009.4322}, 2010.

\bibitem{eppstein}
David Eppstein.
\newblock Maximizing the sum of radii of disjoint balls or disks.
\newblock In {\em Proceedings of the 28th Canadian Conference on Computational
  Geometry, {CCCG} 2016, August 3-5, 2016, Simon Fraser University, Vancouver,
  British Columbia, Canada}, pages 260--265, 2016.

  \bibitem{Erlebach}
Thomas Erlebach, Klaus Jansen and Eike Seidel. \newblock Polynomial-time 
approximation schemes for geometric intersection graphs.
\newblock {\em SIAM J. Comput.}, 34(6):1302--1323,
2005.

\bibitem{FuLY98}
Minyue Fu, Zhi{-}Quan Luo, and Yinyu Ye.
\newblock Approximation algorithms for quadratic programming.
\newblock {\em J. Comb. Optim.}, 2(1):29--50, 1998.

\bibitem{kleinberg2006algorithm}
Jon Kleinberg and Eva Tardos.
\newblock {\em {Algorithm design}}.
\newblock Pearson Education India, 2006.


\bibitem{Megiddo83}
Nimrod Megiddo, \newblock Towards a Genuinely Polynomial Algorithm for Linear Programming.
\newblock SIAM J. Comput. 12(2), pages 347--353, 1983.


\bibitem{Papadimitriou}
Christos~H. Papadimitriou and Kenneth Steiglitz.
\newblock {\em Combinatorial Optimization: Algorithms and Complexity}.
\newblock Prentice-Hall, Inc., Upper Saddle River, NJ, USA, 1982.

\bibitem{PardalosV91}
Panos~M. Pardalos and Stephen~A. Vavasis.
\newblock Quadratic programming with one negative eigenvalue is np-hard.
\newblock {\em J. Global Optimization}, 1(1):15--22, 1991.

\bibitem{Khachiyan}
S.~P. Tarasov, L. G.~Khachiyan, M.~K.~Kozlov.
\newblock The polynomial solvability of convex quadratic programming.
\newblock {\em USSR Computational Mathematics and Mathematical Physics}, 20(5):223 - 228, 1980.



\bibitem{Toth04}
G{\'{a}}bor~Fejes T{\'{o}}th.
\newblock Packing and covering.
\newblock In {\em Handbook of Discrete and Computational Geometry, Second
  Edition.}, pages 25--52. 2004.
  

\bibitem{wiki}
Regular Polygons, \newblock Wikipedia, The Free Encyclopedia. \newblock URL: https://en.wikipedia.org/wiki/Regular\_polygon. 



\end{thebibliography}

\end{document}